\documentclass[11pt]{article}
\usepackage[ascii,latin1,latin2,latin3,latin4,latin5,latin10,latin9,]{inputenc}
\usepackage{indentfirst}
\usepackage{multirow}
\usepackage{mathrsfs}
\usepackage{graphicx}
\usepackage{epstopdf}  
\usepackage{float}
\usepackage{longtable,booktabs,array}
\usepackage{amssymb}
\usepackage{amsmath}
\usepackage{color}
\usepackage{cases}
\usepackage{amsfonts}
\usepackage{mathrsfs}
\usepackage{geometry}
\usepackage{amsthm}
\usepackage{subfigure}
\usepackage{enumerate}
\usepackage{algorithm}
\usepackage{lineno}
\newtheorem{theorem}{Theorem}

\newtheorem{lemma}[theorem]{Lemma}

\newtheorem{remark}{Remark}
\newtheorem{define}{Definition}
\newtheorem{example}{Example}

\def\bref#1{(\ref{#1})}

\def\P{{\bf P}}
\def\pp{{\mathbf{p}}}
\def\qq{{\mathbf{q}}}
\def\Z{{\mathbb{Z}}}
\def\R{{\mathbb{R}}}

\def\V{{\mathbb{V}}}

\def\IR{{\mathbb{I}\mathbb{R}}}

\def\O{{\tilde{\mathcal{O}}}}

\begin{document}
\title{Certified Numerical Real Root Isolation of Zero-dimensional Multivariate Real Nonlinear Systems\footnote{This work was partially supported by the National Key Research and Development Program of China grant 2022YFC3802102.}}

\author{
Jin-San Cheng\footnote{Corresponding author: jcheng@amss.ac.cn}\,\,\, Junyi Wen \\
  KLMM, Institute of Systems Science, Academy of Mathematics and Systems Science,  CAS\\
 School of Mathematical Sciences, University of Chinese Academy of Sciences\\
}

\date{}

\maketitle

\begin{abstract}
Using the local geometrical properties of a given zero-dimensional square multivariate nonlinear system inside a box, we provide a simple but effective and new criterion for the uniqueness and the existence of a real simple zero of the system inside the box.
Based on the result, we design an algorithm based on subdivision and interval arithmetics to isolate all the real zeros of a general real nonlinear system inside a given box. Our method is complete for systems with only finite isolated simple real zeros inside a box. A termination precision is given for general zero-dimensional systems. Multiple zeros of the system are output in bounded boxes. A variety of benchmarks
show the effectivity and efficiency of our implementation (in C++). It works for polynomial systems with Bezout bound more than 100 million. It also works for non-polynomial nonlinear systems. We also discuss the limitations of our method.


\end{abstract}



\section{Introduction}
Real root isolation of equation systems is a fundamental problem in mathematics and engineering applications. There are many famous symbolic computation methods to solve the problem: the Gr\"{o}bner basis method \cite{GB,cox1}, the Ritt-Wu characteristic set method \cite{wu}, the cylindrical algebraic decomposition(CAD)\cite{CAD}  and the resultant method \cite{Gelfand} and so on. Though the size of the polynomial systems that can be solved is limited, symbolic methods can get the algebraic representation(s) of all the complex solutions, even including the multiplicity of the solutions.
Numerical methods, such as the homotopy continuation method \cite{homotopy}, can get all the isolated complex solutions of square polynomial systems (systems with $n$ variables and $n$ polynomials) with large sizes, even for over-determined systems or positive dimensional systems \cite{bertini}.
Though the solutions with traditional homotopy continuation method are lack of certification,
the certified homotopy method was proposed \cite{certifiedhomotopy,robustcertifiedhomotopy,reliablehomotopy,certifiedhomotopyjan,xuchee} in some literatures to overcome the shortcoming based on the famous $\alpha$-theory. To ensure that the output results are reliable, certification with interval Newton's method (see \cite{Krawczyk,Moore,Rump}) for possible real roots after homotopy continuation method is applied in \cite{xiahomoisolation}.

Subdivision method is also used to get the real solutions of the given systems, which are not limited to polynomial case \cite{Burr,Burr2,GS01,Imbach,hannielelber,Kearfott,Mantzaflaris,Mourrain,Neumaier,SP93,Xu}. There are many methods to exclude the domains without solutions \cite{GS01,Mourrain,SP93,Hong0,cheng5}.
More requirements are needed to ensure that the boxes are isolating boxes. There are three main methods for certifying a real root: Miranda theorem with Jacobian test, the interval Newton method, and the $\alpha$-theory.

Miranda theorem \cite{Miranda} and its practical version MK test \cite{GS01, Kioustelidis,Moore2} is used for checking the existence of a real zero of a system inside a box. The MK test works well for linear systems \cite{Moore2}.
 Jacobian test \cite{Aberth} is used for certifying that a system has at most one real zero inside a box, so does the methods in \cite{hannielelber}.  Thus, 
 the termination of the subdivision method based on the MK test and Jacobian test is guaranteed in a theoretical sense for square systems with 
simple roots.
 %
 When MK test is used for systems with more than two variables, it seldom succeeds which can also be found in our experiments. We also analyze the reason in the experiments section.
  Besides that, the interval Newton method \cite{Krawczyk,Moore,Rump} and $\alpha$-theory \cite{Smale} can work for testing the uniqueness and existence of the complex (or real) zeros. The interval Newton method can verify that a box is an isolating box. In \cite{Imbach,Kearfott,Neumaier,Hkapur}, the authors use the interval Newton method for real root finding. But the termination of the method is not ensured for certifying a box containing a simple zero or not by successive subdivision. Since an isolating box of a square system may not satisfy the existence and uniqueness condition, the termination condition of subdivision is absent.   Thus the theory for root isolation of equation systems based on interval Newton method is not complete even for systems with only simple roots. The $\alpha$-theory is slightly different from the other two methods, which computes derivatives with high orders, and the verified domain of $\alpha$-theory is a ball, not a box. 
  Similar as the interval Newton method, the termination of the $\alpha$-theory method for root isolation of a square system is also not complete.  In \cite{Mantzaflaris}, the authors present the concept of the $\alpha$-inclusion box and use it for seeking the real roots of a square system.  In \cite{cheng5}, we presented a new method which was based on the geometrical property, the so-called orthogonal monotonicity inside a box for a bivariate polynomial system to certify the existence and uniqueness of a real root inside the box. We used bounding polynomials to exclude the regions which contained no roots. The termination of the method is guaranteed. Thus it can be used for real root isolation inside a box containing only simple real zeros. The method was extended to bivariate nonlinear systems in the journal version \cite{chengwenJSC}. We extend the method to general zero-dimensional equation systems in this paper.

In this paper, we present a new existence criterion of a simple real root of a zero-dimensional square system inside a box, which is much easier to succeed than Miranda based criterion. Based on the new criterion, we propose an algorithm to isolate the real roots of a zero-dimensional real nonlinear square system $F=(f_1,\ldots,f_n)$ inside a given box $B=[a_1,b_1]\times \ldots \times [a_n,b_n]$, where $f_i$ and $\frac{\partial f_i}{\partial x_j}$ are well defined in $B$ for $1\le i,j\le n$. In order to analyze the roots of the system locally inside a box, we give the concept of the opposite monotone system (O-M system for short) in an $n-$D box and give a criterion to check the existence and uniqueness of a simple real zero of the given system in the given box based on properties of the opposite monotone system. Though the O-M system is firstly presented in \cite{cheng5} for 2-D case (It is called orthogonal monotone in \cite{cheng5}), the O-M condition for $n$-D systems are more complicated than that for 2-D systems.

Simply speaking, for a system $F$ and an $n$-D box $B$, 
we transform locally the original system $F$ into a new system $G=(g_1,\ldots,g_{n})$ with the same zeros as $F$ such that the curve $S=\V(g_1,\ldots,g_{n-1})$ is monotone in $B$ and $S$ intersects transversally with the hypersurface $\V(g_n)$ in $B$ (see Definition \ref{OM}). The evaluation of the functions of the tangent vector of $S$ on $B$ does not contain zeros ensures its monotonicity inside $B$. The direction of the tangent vector of $S$ and the normal vector of $\V(g_n)$ are almost identical or opposite when evaluating on $B$, which ensures the uniqueness of the root inside $B$. 
The existence of a real zero of $G$ inside the box $B$ can be determined by the change of the signs of the evaluations of $g_n$ on the two endpoints of $S$ which are the intersection of $S$ and the boundaries of $B$. Some local transformation techniques in MK test \cite{Kioustelidis,Moore2} are modified and used in our method. The new system in MK test is $J_F(m_B)^{-1}F^T$, here $J_F(m_B)^{-1}$ is the inverse of the Jacobian matrix of $F$ at the middle point of $B$ and $F^T$ is the transpose of $F$. In our method, $G=U\,J_F(m_B)^{-1}F^T$, where $U$ is an invertible matrix under some requirements.
We prove that the termination of the subdivision process for finding all simple real zeros of a system inside a box. So our method is complete for real root isolation of a square nonlinear system inside a bounded box.

Since the existence condition based on the opposite monotone method is used recursively, we revise the original conditions for the opposite monotone system and propose the concept of the strong monotone (S-M) system to avoid constructing opposite monotone systems repeatedly.
For boxes which contain multiple real zeros of systems, our method is invalid, thus we give a terminate precision for subdivision process. Therefore, we may get some suspected boxes which reach the terminate precision and do not satisfy the conditions of our method. 
We give a heuristic verification method to deal with those suspected boxes. Based on our theory, we design an algorithm to isolate the real zeros of a multivariate equation system. We also analyze the complexity of our algorithm.
 We implement our algorithm in C++. Our experiments show the effectivity and efficiency of our method. We compare our method with some existing methods and analyze some aspects of the methods. Notice that our method can be used for complex root isolation since a complex nonlinear system can be transformed into a real nonlinear system. 

The rest of this paper is organized as follows. We introduce some notations and preliminaries in the next section. In Section 3, we give the concepts of O-M system and S-M system in a box and prove the uniqueness and existence theorem, then we show how to match the uniqueness and existence conditions for a given system and a given box. The algorithm of the method is also given in this section. The complexity analysis is also given there. In Section 4,
some experiment results are given and an analysis based on the results is shown. We draw a conclusion in the last section.

\section{Notations and Preliminaries}
In this section, we will give some notations, definitions and basic results.
\subsection{Notations}
Let $C^i(\Omega)$ denote the set of all $i$-order continuous differentiable functions defined in $\Omega$, where $\Omega\subset\R^n$ and $\R$ is the field of real numbers.  Let $B=[a_1,b_1]\times \ldots \times [a_n,b_n]$ be an $n$-D box in $\R^n$. Let
$$F_i^l(B)=\{(p_1,\ldots,p_n)\in B| p_i=a_i\},i=1,\ldots,n,$$
$$F_i^r(B)=\{(p_1,\ldots,p_n)\in B| p_i=b_i\}, i=1,\ldots,n.$$
We call $F_i^l(B)$ or $F_i^r(B)$ a {\bf face} of $B$ for any $1\le i \le n$. Let $m(B)=(\frac{a_1+b_1}{2},\ldots,\frac{a_n+b_n}{2})$ be the middle point of the box $B$ and $w(B)=max\{b_1-a_1,\ldots,b_n-a_n\}$ be the width of $B$. Let $\partial B=\bigcup_{i=1}^n F_i^l(B) \cup \bigcup_{i=1}^n F_i^r(B)$ be the boundaries of the box and $v(B)=\{(p_1,\ldots,p_n)\in B| p_i=a_i \text{\ or \ } b_i, i=1,\ldots,n\}$ be the set of the vertexes of $B$.

Let $F=(f_1(X),\cdots,f_n(X))$ be a function system, where $X=\{x_1,\cdots,x_n\}$ are variables and $f_i(X)\in C^1(B)$. We denote $F=0$ as the equation system $\{f_1=0,\cdots,f_n=0\}$. If a point $\pp\in \R^n$ satisfies $f_1(\pp)=\cdots=f_n(\pp)=0,$ then we call $\pp$ is a real zeros of the function system $F$ or a real root of equation system $F=0$. We denote all the real zeros of $F$ as $\V(F)$. Let $J_F^X$ be the Jacobian matrix of $F$ with respect to $X$ (simply for $J_F$ without misunderstanding). Denote $\IR,\IR^n$ and $\IR^{n\times n}$ by the set of real intervals, $n$-D interval vectors and $n\times n$ interval matrices, respectively. For a (an interval) matrix $M\in \R^{n\times n}$ ($\IR^{n\times n}$), we let $M_{i,:}$ ($M_{:,i}$) denote the $i-$th row (column) of $M$.

Let $f(X)\in C^1(B)$ be a real function and $B$ an $n$-D box. For a subset $B'\subset B$, we denote $f(B')=\{f(\mathbf{p})| \mathbf{p}\in B'\}$ and we say $f(B')>0 (<0)$ if $\forall \mathbf{p}\in B', f(\mathbf{p})>0$ $(<0)$. Similarly, for an interval $I\subset \R$, we say $I>0$ $(<0)$ if $\forall a\in I, a>0$ $(<0)$. We define a sign function of $f(B')$ as following:
$$  \mathrm{Sign}(f(B'))=\left\{
\begin{aligned}
1,&  \text{ if } f(B')>0, \\
-1,& \text{ if } f(B')<0, \\
 0,& \text{ otherwise. }
\end{aligned}
\right.
$$

\subsection{Interval Analysis}
Using our method, we need to compute the evaluation of a function $f(x_1,\ldots,x_n)\in C^1(B)$ on a box $B=I_1\times \cdots \times I_n\subset \R^n$: $f(B)=\{f(\pp)|\pp\in B\}$. However, $f(B)$ is usually difficult to be computed exactly. Interval analysis \cite{Moore0} is a useful tool to compute the enclosure of the range of a function over a box. A real function $f$ can be extended to an {\bf interval function} by interval analysis. The basic arithmetic operations over intervals are as below. Let $I_1=[a,b]\subset \R, I_2=[c,d]\subset \R$.
\begin{eqnarray*}
I_1+I_2&=&[a+c,b+d],\\
I_1-I_2&=&[a-d,b-c],\\
I_1*I_2&=&[\min\{a*c,a*d,b*c,b*d\},\max\{a*c,a*d,b*c,b*d\}],\\
I_1/I_2&=&[a,b]*[1/d,1/c],0\not\in I_2.
\end{eqnarray*}

Let $\Box f$ denote the interval function of $f$, it has two properties \cite{Lien}:
\begin{enumerate}
\item $f(B)\subset \Box f(B)$,
\item ${\lim\limits_{i \to \infty}}\Box f(B_i)=f(\lim\limits_{i \to \infty}B_i)$,
\end{enumerate}
where $B,B_i\subset \R^n$ and $\lim\limits_{i \to \infty}B_i=\pp, \pp$ is a point in $\R^n$.
There are many different forms of $\Box f$,  for polynomial case, a simple way is just using interval arithmetic \cite{Moore0}. Fox example, $g=x^2-x, B=[0,1]$, then $\Box g(B)=[0,1]\cdot [0,1]-[0,1]=[0,1]-[0,1]=[-1,1]$. Since $g(B)=[-\frac{1}{4},0]$, we can find that $\Box g(B)$ is much bigger than $g(B)$.
Notice that most of the real functions (such as exp, sin, cos, etc.) are also easy to be extended to interval functions. 

Since we usually can not get the exact representation of $f(B)$, we use it to represent $\Box f(B)$ for simplicity if there is no doubt in the rest of the paper.

\section{Uniqueness and Existence}
We will give the uniqueness and existence conditions of a square system containing a simple zero inside a box in this section.

Let $\pp\in \R^n$ be an isolated zero of a system $F$. We always assume that the functions in this section are $C^1$ inside the domain we consider. We call $\pp\in \V(F)$ a {\bf simple zero} of $F$ if $\det(J_F(\pp))\ne 0$, otherwise we say $\pp$ is a {\bf singular} or {\bf multiple zero} of $F$.
\subsection{An opposite monotone (O-M) system in a box}
The concept of the opposite monotone system inside a box for $2$-D case is first presented in \cite{cheng5}. We extend it to $n$-D case which is much more complicated. Let $G=(g_1,\ldots,g_n)$ be a nonlinear system and $G'=(g_1,\ldots,g_{n-1})$, our O-M condition is based on the geometric properties of $\V(G')$. Generally speaking, we know that $\V(G')$ is a one-dimensional curve in $\R^n$ if it exists. 
We denote the tangent vector of $\V(G')$ at $\pp\in \V(G')$ as below: 
\begin{equation} \label{eqn-Ti}
{\bf T_p}=(\det(T_1(\pp)),\ldots,(-1)^{i+1}\det(T_i(\pp)),\ldots,(-1)^{n+1}\det(T_n(\pp))),
\end{equation} 
where $X_i=X\backslash x_i=\{x_1,\ldots,x_{i-1},x_{i+1},\ldots,x_n\}$ and $T_i=J_{G'}^{X_i}, i=1,\ldots,n.$

We introduce some definitions and lemmas below and then give the concept of the monotonicity of $S=\V(G')$ in $B$.

\begin{define}
Let ${\bf U}=(u_i)\in \IR^n$  and ${\bf V}=(v_i)\in \IR^n$ be $n$-dimensional interval vectors. We say ${\bf U}$ and ${\bf V}$ are {\bf matched}, if
\begin{enumerate}[(1)]
\item for any $i\in \{1,\ldots,n\}$, $0\notin u_i$ and $0\notin v_i$.
\item for any $i,j \in \{1,\ldots,n\}$, $u_i \cdot v_i$ and $u_j \cdot v_j$ have the same signs.
\end{enumerate}
\end{define}

For example, $([1,2],[2,3],[-2,-1])$ and $([-4,-3],[-2,-1],[1,2])$ are matched. Let ${\bf U_1,U_2,U_3}$ be interval vectors. It is easily to see that the following properties hold.
\begin{enumerate}
\item If ${\bf U_1}$ and ${\bf U_2}$ are matched, ${\bf U_2}$ and ${\bf U_3}$ are matched, then ${\bf U_1}$ and ${\bf U_3}$ are matched.
 \item If ${\bf U_1}$ and ${\bf U_2}$ are matched, then ${\bf U_1}$ and ${\bf -U_2}$ are matched.
\end{enumerate}

The following lemma is well-known and can be found in many text books (see Chap. 2 Part. III in \cite{edwards} for example).
\begin{lemma}\label{MeanValue} (Mean Value Theorem) Let $\Omega\subset \R^n$ be a convex set, $f=f(x_1,\ldots,x_n)$ is a differentiable function defined on $\Omega$. Then $\forall p_1,p_2\in \Omega, \pp_1\ne \pp_2$, $\exists \theta \in (0,1)$ s.t.
$$f(\mathbf{p}_1)-f(\mathbf{p}_2)=\nabla f(\mathbf{p}_2+\theta (\mathbf{p}_1-\mathbf{p}_2))(\mathbf{p}_1-\mathbf{p}_2).$$
\end{lemma}
\begin{define} Let $G'=(g_1,\ldots,g_{n-1})$ and $S=\V(G')=\{\pp\in \R^n| g_1(\pp)=\cdots=g_{n-1}(\pp)=0\}$. We say $S$ is {\bf strong monotonous} in $B$ if $0\notin \det(T_i(B)), i=1,\ldots, n$ (see \bref{eqn-Ti} for $T_i$).
 \end{define}

\begin{example} $G'=(x^2/2+y^2-2*z^2,x^2/2+y^2/2-z^2/2-1/2)$ and $B=[0.10,0.11]\times[0.10,0.11]\times[0.10,0.11]$. We have
$$T_1=J_{G'}^{\{y,z\}}=\left(
                         \begin{array}{cc}
                           2*y & -4*z \\
                           y & -z \\
                         \end{array}
                       \right), T_2=J_{G'}^{\{x,z\}}=\left(
                         \begin{array}{cc}
                           x & -4*z \\
                           x & -z \\
                         \end{array}
                       \right),T_3=J_{G'}^{\{x,y\}}=\left(
                         \begin{array}{cc}
                           x& 2*y  \\
                           x&y  \\
                         \end{array}
                       \right).$$
Thus,
                       \begin{eqnarray*}\det(T_1(B))&=& \det(\left(
                         \begin{array}{cc}
                           2*\mathrm{[0.10,0.11]} & -4*\mathrm{[0.10,0.11]}  \\
                           \mathrm{[0.10,0.11]} & -1*\mathrm{[0.10,0.11]}  \\
                         \end{array}
                       \right) )\\
                       &=&-2*[0.10,0.11]*[0.10,0.11]+4*[0.10,0.11]*[0.10,0.11]\\
                       &=&-2*[0.01,0.0121]+4*[0.01,0.0121]\\
                       &=&[-0.0242,-0.02]+[0.04,0.0484]\\
                       &=&[0.0158,0.0284].
                       \end{eqnarray*}
Similarly, we have $\det(T_2(B))=[0.0279,0.0384], \det(T_3(B))=[-0.0142,-0.0079]$. So $S=\V(G')$ is strong monotonous in $B$ since $0\not\in \det(T_i(B))$ for $i=1,2,3$.
\end{example}

Then, we will prove some nice properties of $S$ if it is strong monotonous in $B$. Let $\pp=(p_1,\cdots,p_n)\in \R^n$, we define $\Pi_i(\pp)=p_i$, for any $i=1,\ldots,n$.
\begin{lemma}\label{monolemma} Let $G'=(g_1,\cdots,g_{n-1})$ and $S=\V(G')$. If $S$ is strong monotonous in $B=[a_1,b_1]\times \cdots \times [a_n,b_n]$, then we have:
 \begin{enumerate}[(a)]
\item $\forall \hat{x_i}\in [a_i,b_i], i=1,\ldots,n$, the hyperplane $x_i=\hat{x_i}$ intersects $S$ at most once in $B$. Moreover, the hyperplane and $S$ are not tangent.
\item $S$ can not be a loop in $B$.
\item $\forall \pp,\qq\in S\cap B$ and $ \pp\ne \qq$, $(\Pi_i(\pp-\qq))_{1\le i \le n}$ and $((-1)^{i+1}\det(T_i(B))_{1\le i \le n}$ are matched.
\end{enumerate}
\begin{proof} (a) We prove only that $\forall \hat{x_1}\in I_1$, the hyperplane $x_1=\hat{x_1}$ intersects $S$ at most once in $B$. The case $i=2,\ldots,n$ can be proved similarly. Assume that the hyperplane $x_1=\hat{x_1}$ intersects $S$ at two points $\pp, \pp'$ in $B$. Let $\pp-\pp'=\Delta x=(\Delta x_1,\Delta x_2,\ldots,\Delta x_n)$. Using mean value theorem for every $g_j$, we have that there exists a point $\qq_j\in B$ s.t.
$$\sum_{i=1}^n \frac{\partial g_j}{\partial x_i}(\qq_j)\Delta x_i=g_j(\pp)-g_j(\pp')=0, j=1,\ldots,n-1.$$
Since $\Delta x_1=0$, thus we have the following equation
\[\begin{pmatrix} 1 &  0 & \cdots & 0 \\ \frac{\partial g_1}{\partial x_1}(\qq_1) &  \frac{\partial g_1}{\partial x_2}(\qq_1) & \cdots &\frac{\partial g_1}{\partial x_{n}}(\qq_1) \\ \vdots & \vdots & \ddots & \vdots \\ \frac{\partial g_{n-1}}{\partial x_1}(\qq_{n-1}) &  \frac{\partial g_{n-1}}{\partial x_2}(\qq_{n-1}) & \cdots &\frac{\partial g_{n-1}}{\partial x_{n}}(\qq_{n-1})\end{pmatrix} \begin{pmatrix} \Delta x_1 \\ \Delta x_2 \\ \vdots \\ \Delta x_n \end{pmatrix}={\bf 0}. \]
Let $A$ denote the above matrix, i.e., $A \Delta x^T={\bf 0}$. Since $\Delta x$ is nonzero, 
it implies $\det(A)=0$. However, since $S$ is strong monotonous in $B$, we have $0\notin \det(T_1(B))$, then we have $\det(A)=1 \det(M_1)\ne 0$, where $M_1\in T_1(B)$. It is a contradiction. Moreover, if $x_1=\hat{x_1}$ and $S$ are tangent at point $\pp$ in $B$, then $T_1(\pp)=0$. It is also a contradiction since $S$ is strong monotonous in $B$.

(b) If $S$ is a loop in $B$, then $\exists \hat{x_i}\in I_i$, the hyperplane $x_i=\hat{x_i}$ must intersect $S$ at two points in $B$. It is a contradiction with (a).

(c) For any two point $\pp, \qq$, let $\pp-\qq=\Delta x=(\Delta x_1,\Delta x_2,\ldots,\Delta x_n)$. By Lemma \ref{monolemma} (a), we know that $\Delta x_i\ne 0, \forall i=1,\ldots,n$. Then, we need only to prove that for $i=2,\ldots,n$, $\Delta x_1 \det(T_1(B))$ and $\Delta x_i (-1)^{i+1}\det(T_i(B))$ have the same signs.

We consider the following equation system
\[\begin{pmatrix} \Delta x_i & 0 &\cdots & 0 & -\Delta x_1 & 0 & \cdots & 0 \\ \frac{\partial g_1}{\partial x_1}(\qq_1) & \cdots & \cdots & \cdots & \frac{\partial g_1}{\partial x_i}(\qq_1) & \cdots & \cdots &\frac{\partial g_1}{\partial x_{n}}(\qq_1) \\ \vdots & \vdots & \vdots & \vdots & \vdots & \vdots & \vdots & \vdots \\ \frac{\partial g_{n-1}}{\partial x_1}(\qq_{n-1}) & \cdots & \cdots & \cdots &\frac{\partial g_{n-1}}{\partial x_i}(\qq_{n-1}) & \cdots &\cdots & \frac{\partial g_{n-1}}{\partial x_n}(\qq_{n-1}) \end{pmatrix} \begin{pmatrix} \Delta x_1 \\ \Delta x_2 \\ \vdots \\ \Delta x_n \end{pmatrix}={\bf 0}. \]
Let $A'$ denote the above matrix, whose last $n-1$ rows are the same as $A$, we have
\[\det(A')=\Delta x_i \det(M_1)- \Delta x_1 (-1)^{i+1} \det(M_i),\]
where $M_1\in T_1(B), M_i\in T_i(B)$.
If \[\Delta x_1 \det(T_1(B)) \cdot \Delta x_i (-1)^{i+1}\det(T_i(B))<0,\]
we have $\Delta x_i \det(M_1)$ and $\Delta x_1 (-1)^{i+1} \det(M_i)$ have different signs, then $\det(A')\ne 0$, and $\Delta x={\bf 0}$. It is a contradiction with $\pp,\qq$ are two different points.
Therefore,
\[\Delta x_1 \det(T_1(B)) \cdot \Delta x_i (-1)^{i+1}\det(T_i(B))>0.\]
Notice that $\Delta x_i$ and $\det(M_i)$ are all nonzero for $i=1,\ldots,n$.
Then we know that for any $j=1,\ldots,n$, $\Delta x_j (-1)^{j+1}\det(T_j(B))$ have the same signs, i.e., $(\Pi_i(\pp-\qq))_{1\le i \le n}$ and $((-1)^{i+1}\det(T_i(B))_{1\le i \le n}$ are matched.
\end{proof}
\end{lemma}

Now, we will give the definition of an O-M system in a box $B$.

\begin{define}\label{OMdef} Let $M\in \IR^{n\times n}$. For any $i=1,\ldots,n$, we denote $M_{\{n,i\}}\in \IR^{(n-1) \times (n-1)}$ as a sub-matrix of $M$ by deleting the $n$-th row and $i$-th column and we denote $M_{n,:}$ as $n-$th row of $M$.  We say $M$ is an O-M matrix if $((-1)^{i+n}\det(M_{\{n,i\}}))_{1\le i\le n}$ and $M_{n,:}$ are matched.
\end{define}
\begin{remark} If $M\in \R^{n\times n}$ is a matrix and satisfies the above conditions, we also regard $M$ as an O-M matrix.
\end{remark}

\begin{example} Let $M=\begin{pmatrix} [1,2] & [3,4] & [-1,1] \\ [3,4] & [-1,1] & [5,6] \\ [1,2] & [-2,-1] & [-2,-1]  \end{pmatrix}$. We have
\begin{eqnarray*}
(-1)^{1+3}\det(M_{\{3,1\}})&=&[14,25],\\
(-1)^{2+3}\det(M_{\{3,2\}})&=&[-16,-1],\\
(-1)^{3+3}\det(M_{\{3,3\}})&=&[-18,-7],
\end{eqnarray*}
and
\[M_{3,:}=([1,2], [-2,-1], [-2,-1]).\]
Therefore, $M$ is an O-M matrix.
\end{example}

\begin{define} Let $F=(f_1,\ldots,f_n)$, $F'=(f_1,\ldots,f_{n-1})$ and $S=\V(F')$. We say $F$ is an {\bf O-M system} in $B$ if
\begin{enumerate}[(1)]
\item $S$ is strong monotonous in $B$.
\item $((-1)^{i+n} T_i(B))_{1\le i\le n}$ and $(\frac{\partial f_n}{\partial x_i}(B))_{1\le i\le n}$ are matched, where $T_i=J_{F'}^{X_i}$ (see \bref{eqn-Ti} for more details).
\end{enumerate}
i.e, $J_F^X(B)$ is an O-M matrix.
\end{define}
We can use interval evaluation to get $J_F^X(B)$ and check whether it is an O-M matrix or not.

\begin{theorem}\label{OM} Let $F=(f_1,\ldots,f_n)$ and $B$ a box. If $F$ is an O-M system in $B$, then $F$ has at most one zero in $B$.
\begin{proof} Assume that $F$ has two different real zeros $\pp=(p_1,\ldots,p_n), \qq=(q_1,\ldots,q_n)$ in $B$. Using Mean Value Theorem for $f_n$, there exists a point $\pp'\in B$ such that
\[\sum_{i=1}^n\frac{\partial f_n}{\partial x_i}(\pp') \cdot \Pi_i(\pp-\qq)=f_n(\pp)-f_n(\qq)=0.\]

However, by Lemma \ref{monolemma}, since $S$ is strong monotonous in $B$, we have $(\Pi_i(\pp-\qq))_{1\le i\le n}$ and $((-1)^{i+1}\det(T_i(B))_{1\le i\le n}$ are matched. By the definition of O-M system, we have $((-1)^{i+n} T_i(B))_{1\le i\le n}$ and $(\frac{\partial f_n}{\partial x_i}(B))_{1\le i\le n}$ are matched. Therefore, $(\Pi_i(\pp-\qq))_{1\le i\le n}$ and $(\frac{\partial f_n}{\partial x_i}(B))_{1\le i\le n}$ are matched. Then, we know that
\[\forall \qq'\in B, \sum_{i=1}^n\frac{\partial f_n}{\partial x_i}(\qq') \cdot \Pi_i(\pp-\qq)\ne 0.\]
It is a contradiction. Hence $F$ has at most one zero in $B$.
\end{proof}
\end{theorem}

\subsection{Preconditioner}
In this subsection, we will find an equivalent system of the original system inside a box, that is, two systems have the same solution(s) inside the box, such that the new system satisfies the O-M condition inside the box.

We give the following lemma first. Though it is clear, we give the proof below.
\begin{lemma}\label{transformsys} Let $F=(f_1,\cdots,f_n)$ and $M\in \R^{n\times n}$ be an $n\times n$ invertible matrix. Then $\V(F)=\V(MF^T)$, where $F^T$ is the transpose of $F$.
\begin{proof} On one hand, $\forall p\in \V(F)$, we have $F(\pp)=0$, then $MF^T(\pp)=0$ i.e., $\pp\in \V(MF^T)$.
On the other hand, $\forall p\in \V(MF^T)$, we have $MF^T(\pp)=0$. Since $M$ is an invertible matrix, then we have $F^T(\pp)=M^{-1}MF^T(\pp)=0$ i.e., $\pp\in \V(F)$. Thus, we have $\V(F)=\V(MF^T)$.
\end{proof}
\end{lemma}

We give the preconditioner which transforms locally a square system $F$ into a new system $U J_F^{-1}(\pp)\cdot F^T$, where $U\in \R^{n\times n}$ is an O-M matrix (In the rest of the paper, $U$ always denotes an O-M matrix) and $\pp\in B$. In general, we choose $\pp$ as $m(B)$. The idea of multiplying $J_F^{-1}(\pp)$ to the original system originates from \cite{Kioustelidis} and used by \cite{GS,Moore2,Mourrain}.
They just transform locally the system $F$ into $J_F^{-1}(\pp)\cdot F^T=(\tilde{f_1},\ldots,\tilde{f_n})^T$ s.t. $\V(\tilde{f_i})$ are almost orthogonal to each other in the neighborhood of $\pp$. Then the Miranda theorem can be used to check the existence of a real zero of the system. But we further do a rotation on $J_F^{-1}(\pp)\cdot F^T$ by multiplying the matrix $U$ to make the system $U J_F^{-1}(\pp)\cdot F^T$ becoming an O-M system in $B$. See the example below for illustration. A similar example to illustrate the same problem can be found in \cite{cheng5}.

\begin{example}
\label{ex3} Let $F=(2\,{y}^{2}-{z}^{2}+2\,x+ 0.16,{x}^{2}+x+3\,y-z- 0.02,3\,{x}^{2}- 5.08 \,x- 0.2492-4\,{y}^{2}-6\,z+3\,{z}^{2}+3\,y )$, $B=[-0.09,-0.04]\times[0.01,0.06]\times[0.01,0.06], U=\begin{pmatrix} 3 & 1 &1 \\ 1 & -3&1\\1&1&-3 \end{pmatrix}$. $F$ has a unique zero $(-0.080966,0.049827, 0.055071)$ in $B$. However, we can find that $F$ is not an O-M system in $B$. The preconditioner in \cite{Kioustelidis} produces  $G'=J_F^{-1}(m(B))\cdot F^T= \small{\left[ \begin {array}{c}  1.09864\,{y}^{2}- 0.570758\,{z}^{2}+
 1.11937\,x+ 0.0864214- 0.0667392\,{x}^{2}+ 0.131059\,z- 0.0715789\,y
\\ \noalign{\medskip}- 2.42325\,{y}^{2}+ 1.62329\,{z}^{2}- 2.93426\,x-
 0.163387+ 1.16858\,{x}^{2}- 2.40357\,z+ 1.03577\,y
\\ \noalign{\medskip}- 2.02830\,{y}^{2}+ 1.22491\,{z}^{2}- 2.45787\,x-
 0.143304+ 0.430319\,{x}^{2}- 1.06258\,z+ 0.026427\,y\end {array}
 \right]
}
$, but our preconditioner generates $G=U J_F^{-1}(m(B))\cdot F^T$ $= \\ \small{\left[ \begin {array}{c} - 1.15563\,{y}^{2}+ 1.13593\,{z}^{2}-
 2.03402\,x- 0.047427+ 1.39868\,{x}^{2}- 3.07297\,z+ 0.847460\,y
\\ \noalign{\medskip} 6.34009\,{y}^{2}- 4.21572\,{z}^{2}+ 7.46428\,x+
 0.433278- 3.14216\,{x}^{2}+ 6.27919\,z- 3.15246\,y
\\ \noalign{\medskip} 4.76029\,{y}^{2}- 2.62220\,{z}^{2}+ 5.55872\,x+
 0.352946- 0.18912\,{x}^{2}+ 0.91523\,z+ 0.884910\,y\end {array}
 \right]}$.  We can find that both $F$ and $G'$ do not satisfy the condition of Miranda theorem. But it is easy to check that $G$ is an O-M system in $B$.
\begin{figure}[htbp]
\centering
\begin{minipage}{0.3\textwidth}
\centering
\includegraphics[scale=0.35]{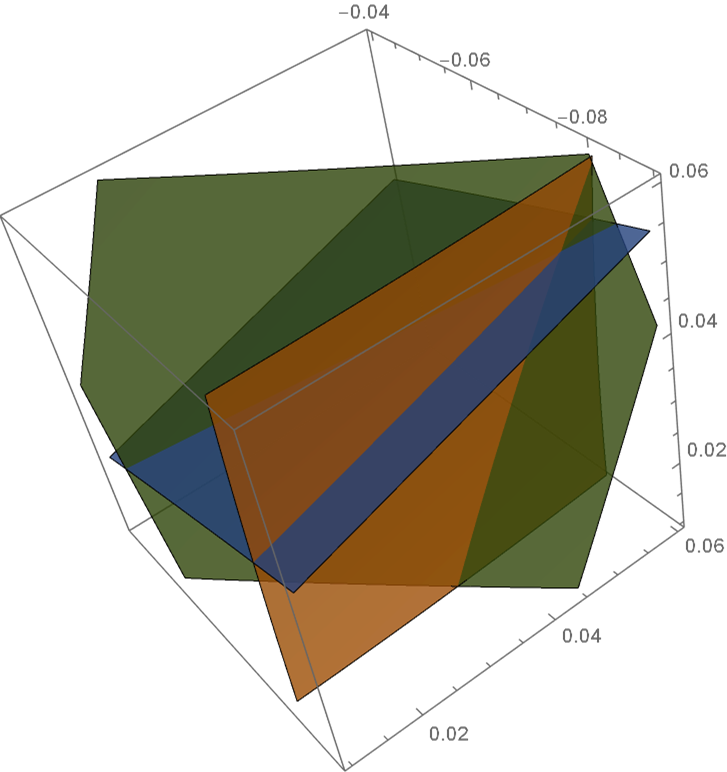}
\end{minipage}
\centering
\begin{minipage}{0.3\textwidth}
\centering
\includegraphics[scale=0.35]{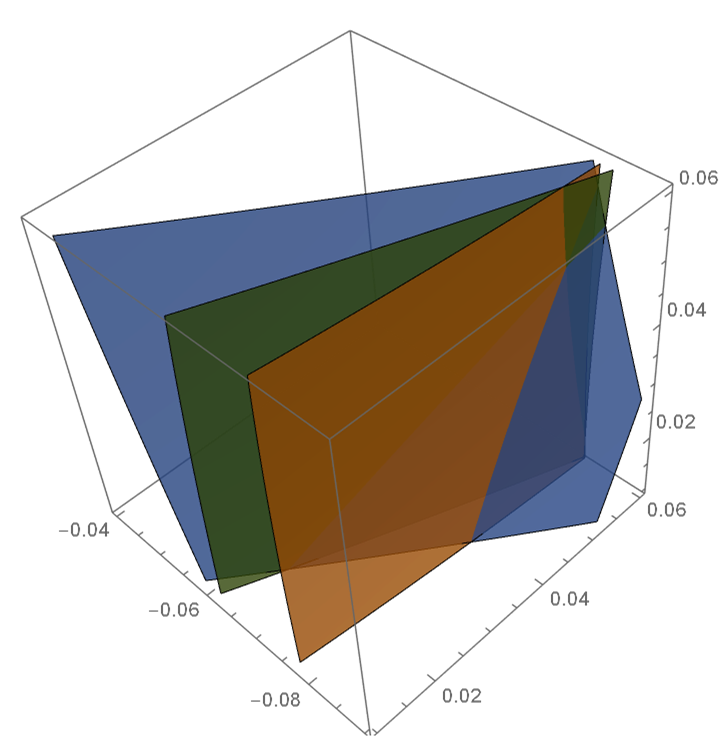} 
\end{minipage}
\centering
\begin{minipage}{0.3\textwidth}
\centering
\includegraphics[scale=0.35]{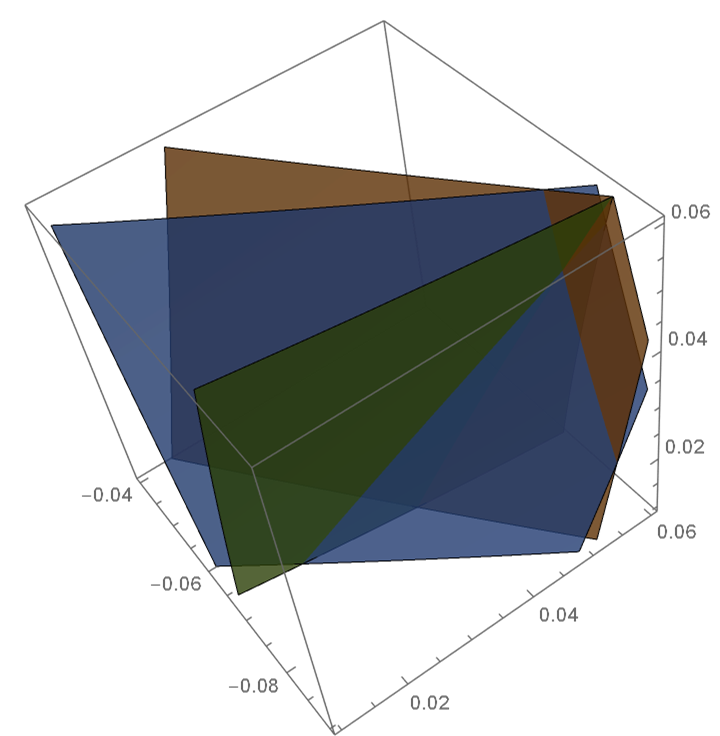} 
\end{minipage}
\caption{The left figure is of $F=0$, the middle one is of $G'=0$, the right one is of $G=0$.} \label{fig-trans}
\end{figure}
\end{example}

\noindent{\bf Remark}: We have many choices for the O-M matrix $U$. For example, when $n=2$, $\begin{pmatrix} 1 & 1 \\ 1 & -1 \end{pmatrix}$ and $\begin{pmatrix} 2 & 1 \\ 1 & -2 \end{pmatrix}$ are both O-M matrices. When $n=3$, $\begin{pmatrix} 2 & 1 & 0 \\ 1 & 2 &  1\\ 1 & -1 & 2 \end{pmatrix}$, $\begin{pmatrix} 3 & 1 & 1 \\ 1 & -3 &  1\\ 1 & 1 & 3 \end{pmatrix}$ and $\begin{pmatrix} 3 & 1 & 1 \\ 1 & -3 &  1\\ 1 & 1 & -3 \end{pmatrix}$ are all O-M matrices. We will discuss how to select a better rotation matrix later.

For each simple zero $\pp^*$ of $F$, we will prove that there always exists a small box $B$ containing $\pp^*$ s.t. $U J_F^{-1}(m(B))\cdot F^T$ is an O-M system in $B$.

For a point $\pp\in \R^n$ and a positive number $\delta>0$, we define a set of box as $B(\pp,\delta)=\{B | B$ is a box and $\pp\in B, w(B)<\delta \}$. Then, we have the following theorem:
\begin{theorem}\label{preconditioner} Let $F=(f_1,\ldots,f_n), f_i\in C^1(B), 1\le i\le n$ and $\pp^*$ a simple zero of $F$. Then, $\exists \delta>0$ s.t. $\forall B\in B(\pp^*,\delta)$, $U J_F^{-1}(m(B))\cdot F^T$ is an O-M system in $B$.
\begin{proof}  Let $G=U J_F^{-1}(m(B))\cdot F^T$, when $B\to \pp^*$, we have
$\lim \limits_{B\to \pp^*} J_{G}(B)=\lim \limits_{B\to \pp^*} U J_F^{-1}(m(B))\cdot J_F(B) =U J_F^{-1}(\pp^*)\cdot J_F(\pp^*)=U$. Since $U$ is an O-M matrix, we prove the theorem.
\end{proof}
\end{theorem}
\noindent{\bf Remark}: If $\pp^*$ is a singular root, then for any box $B$ containing $\pp^*$, we can not transform $F$ into an O-M system in $B$ since $0=\det(J_F(\pp^*))\in \det(J_F(B))$. Hence, our method is invalid for singular roots.

 For an O-M system in $B$, it has at most one real zero in $B$ by Theorem \ref{OM}. We will discuss how to determine whether $F$ does have a zero in $B$ in the next subsection.

\subsection{Existence}
In this subsection, we will give a method to determine whether an O-M system has one zero inside a box or not.

%


\begin{lemma}\label{monolemma2} Let $G'=(g_1,\cdots,g_{n-1})$ and $S=\V(G')$. If $S$ is strong monotonous in $B=[a_1,b_1]\times \cdots \times [a_n,b_n]$, then we have that $B$ contains at most one connected component of $S$.
\begin{proof}  Assume that $B$ contains two connected components of $S$: $L,L'$. By Lemma \ref{monolemma} (a) and (b), we know that $L$ intersects the boundaries of $B$ at most twice. We prove the lemma in the following two cases:

1. If $L$ intersects the boundaries of $B$ twice, there are two cases:

\quad (1.1) $L$ passes $B$ from $F_i^l(B)$ to $F_i^r(B)$. For a point $\qq\in L'$, we have $\pi_i(\qq)\in [a_i,b_i]$. Since $L$ is a connected component, $\exists \pp\in L$ such that $\Pi_i(\pp)=\Pi_i(\qq)$. It is a contradiction with Lemma \ref{monolemma}.

\quad (1.2) $L$ passes $B$ from $F_i^l(B)$ to $F_j^r(B), i\ne j$. Let $L\cap F_i^l(B)=\{\pp_1\}$ and $L\cap F_j^r(B)=\{\pp_2\}$, then we have $\Pi_i(\pp_1)=a_i$ and $\Pi_j(\pp_2)=b_j$. For a point $\qq\in L'$, since $L$ is a connected component and Lemma \ref{monolemma} (a), we have $\Pi_i(\qq)\notin[a_i,\Pi_i(\pp_2)]$ i.e., $\Pi_i(\qq)>\Pi_i(\pp_2)$ and obviously $\Pi_j(\qq)<b_j=\Pi_j(\pp_2)$. Then we have
$$\Delta x_i=\Pi_i(\pp_1)-\Pi_i(\pp_2)<0, \Delta x_j=\Pi_j(\pp_1)-\Pi_j(\pp_2)<0,$$
$$\Delta x'_i=\Pi_i(\qq)-\Pi_i(\pp_2)>0, \Delta x'_j=\Pi_j(\qq)-\Pi_j(\pp_2)<0.$$
 Thus we have $\Delta x_i\cdot \Delta x_j>0$ and $\Delta x'_i\cdot \Delta x'_j<0$.  However, by Lemma \ref{monolemma} (c), $\forall \pp\ne \qq\in S\cap B$, $\Pi_i(\pp-\qq) (-1)^{i+1} \det(T_i(B)) \Pi_j(\pp-\qq) (-1)^{j+1}\det(T_j(B))>0$. Since the sign of $(-1)^{j+1} \det(T_i(B))$ $(-1)^{j+1}\det(T_j(B))$ is unchanged, we know the sign of $\Pi_i(\pp-\qq) \Pi_j(\pp-\qq)$ is unchanged too. Therefore we get a contradiction. Notice that $L$ passes $B$ from $F_i^r(B)$ to $F_j^l(B)$ ($F_i^l(B)$ to $F_j^l(B)$, $F_i^r(B)$ to $F_j^r(B)$) can be proved in a similar way.

2. If $L$ and $L'$ both intersect the boundaries of $B$ only once,  we know that the intersection points must be vertexes of $B$. Let $L\cap\partial B=\{\pp\}, L'\cap\partial B=\{\pp'\}$.  Without loss of generality, assume that $\pp=(a_1,\ldots,a_n)$, then by Lemma \ref{monolemma} (a), $\pp'$ must be $(b_1,\ldots,b_n)$. Considering the tangent vector of $L$ at $\pp$, by Lemma \ref{monolemma} (c), we have $(\Pi_i(\pp-\pp'))$ and $((-1)^{i+1}\det(T_i(B))$ are matched, therefore the sign of the tangent vector $(Sign((-1)^{i+1}\det(T_i(\pp)))=(1,1,\ldots,1)$ or $(-1,-1,\ldots,-1)$. That is to say $L$ will go forward to the inside of $B$. Then $L$ must go out of $B$ and intersects the boundaries of $B$ at another point. It is contradiction with the assumption. Notice that if $\pp$ is another vertex, we can prove it in a similar way. Thus, we prove the lemma.
\end{proof}
\end{lemma}

By Lemma \ref{monolemma}, we know that if $S$ is strong monotonous in $B$, then $S$ intersects the boundaries of the box $\partial B$ at most twice, that is to say $\#(S\cap \partial B)\leq 2$, where $\#(A)$ denotes the number of elements of a set $A$. Next, we determine whether $F$ has a root or not in $B$ based on the three cases: $\#(S\cap \partial B)=0,1,2$.

If $\#(S\cap \partial B)=0$, since $S$ can not be a loop in $B$, then we know that $S$ dose not pass the box $B$. If $\#(S\cap \partial B)=1$, we know that $S$ intersects only the box at a point and this point must be a vertex of $B$. Next, we will analyze the last case: $\#(S\cap \partial B)=2$.

\begin{lemma}\label{case2} Let $F=(f_1,\ldots,f_n)$ be an O-M system in $B$. Let $F'=(f_1,\ldots,f_{n-1})$ and $S=\V(F')$. Assume that $S\cap \partial B=\{\pp_1,\pp_2\}$, we have:
  \begin{enumerate}[(1)]
\item If $f_n(\pp_1)f_n(\pp_2)\leqslant 0$, $F=0$ has a unique root in $B$.
\item If $f_n(\pp_1)f_n(\pp_2)> 0$, $F=0$ has no root in $B$.
\end{enumerate}
\begin{proof} Since $B$ contains a unique component of $S$, we can parameterize it as $\eta(t)=(x_1(t),\ldots,x_n(t))$, $t_1\leqslant t \leqslant t_2$, where $\eta(t_i)=\pp_i, i=1,2$. Since $S$ is strong monotone inside $B$, $\eta(t) (t_1\le t\le t_2)$ is strictly contained inside $B$. Consider the univariate function $g(t)=f_n(x_1(t),\ldots,x_n(t))$. If $f_n(\pp_1)f_n(\pp_2)\leqslant 0$, i.e. $g(t_1)g(t_2)\leqslant 0$, then $\exists t'\in [t_1,t_2]$ s.t. $g(t')=0$. Thus $(x_1(t'),\ldots,x_n(t'))$ is a root of $F=0$. Since $F$ is an O-M system in $B$, we know that $F=0$ has a unique root in $B$. If $f_n(\pp_1)f_n(\pp_2)>0$, i.e., $g(t_1)g(t_2)> 0$, then there are even number roots (counting multiplicity) of $g(t)=0$ in $[t_1,t_2]$. Since $F=0$ has most one root in $B$, we can know that $F=0$ has no root in $B$.
\end{proof}
\end{lemma}

\subsection{Checking the existence}
We will show how to check the existence of a real zero of a system in a box in this subsection. 

Let $F=(f_1,\ldots,f_n)$ be an O-M system in $B$ and $F'=(f_1,\ldots,f_{n-1})$. By the definition of the O-M system, we know that $S=\V(f_1,\ldots,f_{n-1})$ is strong monotonous in $B$. Then by Lemma \ref{case2}, in order to check the existence, we need to compute the intersection points $\pp_1,\pp_2$ of $S$ and $\partial{B}$. However, it is not easy and unnecessary to get the points exactly, we need only to get two $(n-1)$-boxes containing $\pp_1, \pp_2$, say $B_{\pp_1},B_{\pp_2}$. Then by computing $ \mathrm{Sign}(f_n(B_{\pp_1})),  \mathrm{Sign}(f_n(B_{\pp_2}))$, we can get $\mathrm{Sign}(f_n(\pp_1)),  \mathrm{Sign}(f_n(\pp_2))$.  The signs of $f_n(B_{\pp_1}), f_n(B_{\pp_2})$ can be obtained by interval computation, since
$0\notin \frac{\partial f_n}{\partial x_i}(B)$, we can use the following lemma to compute the signs in an easier way.

\begin{lemma}\label{vertex} Let $f(x_1,\ldots,x_n)\in C^1(B)$ and $B=[a_1,b_1]\times \cdots \times [a_n,b_n]$ a box. If $0\notin \frac{\partial f}{\partial x_i}(B), i=1,\ldots,n$ and $f(v(B))>0(<0)$, then $f(B)>0(<0)$.
\begin{proof} WLOG, we assume that $f(v(B))>0$ since $f(v(B))<0$ can be proved similarly. We prove the lemma with mathematical induction:

We first prove the case $n=1$, i.e., $B=[a_1,b_1]$: Since $\frac{\partial f}{\partial x_1}(B)>0$, we know that $f$ is a monotonous univariate function, then we have $f(B)>0$ since $f(a_1)>0$ and $f(b_1)>0$.

Next we assume that the lemma is proved for $n=k-1$ and we are going to prove the case $n=k$, where $k$ is a positive integer. Let $B_i=[a_1,b_1]\times \cdots \times [a_{i-1},b_{i-1}] \times [a_{i+1},b_{i+1}] \times \cdots \times [a_n,b_n]\subset \R^{n-1}, i=1,\ldots,n$ and $t_i=a_i$ or $b_i$. $\forall 1\le i\le n$, since $f|_{x_i=t_i}(v(B_i))\subset f(v(B))$, we have $f|_{x_i=t_i}(v(B_i))>0$. Then by the assumption we have $f|_{x_i=t_i}(B_i)>0$ since the function $f|_{x_i=t_i}$ is with $k-1$ variables. In summary, we have $f(\partial(B))>0$.
$\forall p_1\in [a_1,b_1]$, we have $f|_{x_1=p_1}(v(B_1))\subset f(\partial(B))$, thus $f|_{x_1=p_1}(v(B_1))>0$. Then by the assumption we have $f|_{x_1=p_1}(B_1)>0$. By the arbitrariness of $p_1$, we have $f(B)>0$.
\end{proof}
\end{lemma}

By Lemma \ref{vertex},  we need only to compute the signs of $f_n$ at the vertexes of $B_{\pp_1},B_{\pp_2}$. If $f_n(v(B_{\pp_1}))$, $f_n(v(B_{\pp_2}))$ are all positive or negative, then we know the signs of $f_n(B_{\pp_1}), f_n(B_{\pp_2})$. So do the signs of $f_n(\pp_1),f_n(\pp_2)$.

Next, we will show how to check the existence of a real zero of the system in $B$. Let $B=[a_1,b_1]\times \cdots \times [a_n,b_n]$ and $B_i=[a_1,b_1]\times \cdots \times [a_{i-1},b_{i-1}] \times [a_{i+1},b_{i+1}] \times \cdots \times [a_n,b_n]\subset \R^{n-1}, i=1,\ldots,n$. For all $p_i\in [a_i,b_i]$, let $F'|_{x_i=p_i}=(f_1|_{x_i=p_i}, \ldots,f_{n-1}|_{x_i=p_i})=$ $(f_1(x_1,\ldots,p_i,\ldots,x_n), \ldots$, $f_{n-1}(x_1,\ldots,p_i,\ldots,x_n))$. For convenience, we define the index set of the box $B$ with $2n$ elements as follows
$$Ind(B)=\{(i,a_i), (i,b_i) | i=1,\ldots,n\}.$$
If $S\cap \partial B=\{\pp_1, \pp_2\}$, it is easy to see that there exist $(i_1, t_1), (i_2, t_2)\in Ind(B)$ such that $F'|_{x_{i_1}=t_1}$, $F'|_{x_{i_2}=t_2}$ both have a unique zero in $B_{i_1}, B_{i_2}$ respectively. Therefore, for $(i, t)\in Ind(B)$, we need to know the square system $F'|_{x_i=t}$ has a zero in $B_i$ or not. We can solve those systems by our method recursively. Finally we need to isolate some bivariate systems, which is solved in \cite{cheng5}. Notice that for non-polynomial case, the method is also valid.  We give the algorithm structure of checking existence as following:

 \begin{enumerate}[(1)]
 \item Solve systems $F'|_{x_i=t}$ in $B_i$ where $(i, t)\in Ind(B)$.
 \item If $\forall (i, t)\in Ind(B)$, $F'|_{x_i=t}$ has no zero in $B_{i}$, it means that $S$ does not pass the box $B$ and $F$ has no zero in $B$.
 \item Else if $\exists (i_1, t_1), (i_2, t_2)\in Ind(B)$ such that both $F'|_{x_{i_1}=t_1}$ and $F'|_{x_{i_2}=t_2}$ have a unique zero in $B_{i_1}, B_{i_2}$ respectively. Then we compute the signs of $f_n|_{x_{i_1}=t_1}(v(B_{i_1}))$ and $f_n|_{x_{i_2}=t_2}(v(B_{i_2}))$.
      \begin{enumerate}[(a)]
      \item If $ \mathrm{Sign}(f_n|_{x_{i_1}=t_1}(v(B_{i_1})))  \mathrm{Sign}(f_n|_{x_{i_2}=t_2}(v(B_{i_2})))<0$, $F$ has a unique zero in $B$.
     \item Else, if $ \mathrm{Sign}(f_n|_{x_{i_1}=t_1}(v(B_{i_1})))  \mathrm{Sign}(f_n|_{x_{i_2}=t_2}(v(B_{i_2})))>0$, $F$ has no zero in $B$.
     \item Else if we can not determine the sign of $f_n|_{x_{i_1}=t_1}(v(B_{i_1}))$ or $f_n|_{x_{i_2}=t_2}(v(B_{i_2})$, we need to refine the boxes $B_{i_1},B_{i_2}$ and check the signs again.
      \end{enumerate}
 \item Else, we can not determine if $F$ has a real zero in $B$ or not. We need to subdivide the box and check again.
 \end{enumerate}

In Step(1), we will solve $2n$ systems with $n-1$ equations and $n-1$ variables in the worst case, that is to say, we need to solve $2n\times 2(n-1)\times \cdots \times 2\cdot 2=2^{n-1}\cdot n!$ univariate functions recursively in the worst case.

In Step(3-c), if we can not determine the signs of $f_n|_{x_{i_j}=t_j}(v(B_{i_j})), j=1,2$, we will refine the box $B_{i_j}$. That is to say, we will subdivide $B_{i_j}$ and find the unique box containing the zero of $F'|_{x_{i_j}=t_j}$ in those sub-boxes.  The worst case is that the width of refined box is very small, but we still can not determine the signs. This case happens when the zero $\pp^*$ of $F$ is on (or very close to) the boundaries of $B$, we will discuss the bad situation later.

Based on the above discussion, if a system and a box satisfy our O-M and existence conditions, we can determine that the box is an isolating box of the system. We can design an algorithm for real root isolation based on the O-M condition and the subdivision method. However, from the discussion above, we know that a recursive solving is required in our existence checking step. We need to ensure that the related system in the related box is an O-M system. In order to avoid constructing an O-M system for each one by one, we prefer to construct all these systems in their related boxes together.
In doing so, we introduce the Jacobian test.  One can find that the O-M condition implies the famous Jacobian test based on the following theorem. One can find it in many places, see the corollary of Theorem 12.1 in \cite{Aberth}.

\begin{theorem}\label{Jaco_Con} Let $G=(g_1,\ldots,g_n), g_i\in C^1(B)$. If $0\notin \det(J_G^X(B))$, then $G=0$ has at most one real root in $B$.
\end{theorem}

It is clear that the Jacobian test is weaker than our conditions in Definition \ref{OM}. The Condition (1) in Definition \ref{OM} is for existence checking. The Condition (2) in Definition \ref{OM} is for uniqueness checking, but it is not necessary. We can replace it with a weaker condition as the Jacobian test. In fact, when we check the existence, for $(i,t)\in Ind(B)$,  we need to know whether the system $F'|_{x_i=t}$ has a zero in $B_{i}$ or not. In some cases, the system $F'|_{x_i=t}$ may not be an O-M system in the related box, thus we can not use the existence condition directly and need an $(n-1)$-D O-M matrix $U$ to help us construct O-M system again. This operation may take much time, which is not necessary. In order to improve the situation, we present the Strong Monotone (S-M) system revised from O-M system.

\subsection{Strong Monotone system (S-M)}

Giving an interval matrix $M\in \IR^{n\times n}$, we denote $M_i\in \IR^{i\times n} (1\le i\le n)$ as the first $i$ rows of $M$ and denote $S(M_i)$ as the set of all $i$-order sub-matrices of $M_i$. We give the following definition:

\begin{define}\label{OMdef} Let $M\in \IR^{n\times n}$ and $M=(m_{ij})_{1\le i,j\le n}$.  We say $M$ is an S-M matrix if
\begin{itemize}
\item $0\notin \det(A), \forall A\in S(M_i)$, $1\le i\le n$.
\end{itemize}
\end{define}

\begin{remark} If $M\in \R^{n\times n}$ is a matrix and satisfy the above conditions, we also regard $M$ as an S-M matrix.
\end{remark}

\begin{example} Let $M=(m_{ij})=\begin{pmatrix} [3,4] & [1,2] & [1,2] \\ [1,2] & [3,4] & [-2,-1] \\ [1,2] & [-2,-1] & [-2,-1]  \end{pmatrix}$. We can verify that $M$ is an S-M matrix. First, we have $0\notin m_{1j}$, $1\le j\le 3$. Then for $i=2$, we have that $M_2=\begin{pmatrix} [3,4] & [1,2] & [1,2] \\ [1,2] & [3,4] & [-2,-1]\end{pmatrix}$ and $S(M_2)=\{A_1,A_2,A_3\}=\{\begin{pmatrix} [1,2] & [1,2] \\ [3,4] & [-2,-1] \end{pmatrix}, \begin{pmatrix} [3,4] & [1,2] \\ [1,2] & [-2,-1] \end{pmatrix}$, $\begin{pmatrix} [3,4] & [1,2] \\ [1,2] & [3,4] \end{pmatrix} \}$. After computing, we have $\det(A_1)=[-12,-4],\det(A_2)=[-12,-4]$ and $\det(A_3)=[5,15]$. For $i=3,$ we have $M_3=M$, $S(M_3)=\{M\}$ and $\det(M)=[-72,-16]$. Therefore, $M$ is an S-M matrix.
\end{example}
In the following, we always denote $V$ as an S-M matrix, then we will introduce the definition of the S-M system.
\begin{define}\label{SM} Let $F=(f_1,\ldots,f_n), f_i\in C^1(B), 1\le i\le n$ and $B$ a box. We call $F$ an {\bf S-M system} in $B$ if $J_F^X(B)$ is an S-M matrix.
\end{define}

\begin{theorem}\label{thm-OM} Let $F=(f_1,\ldots,f_n), f_i\in C^1(B), 1\le i\le n$. If $F$ is an S-M system in $B$, then $F=0$ has at most one root in $B$.
\begin{proof} Since $F$ is an S-M system in $B$, we have $0\notin \det(J_F^X(B))$. Then by Theorem \ref{Jaco_Con}, we know that $F=0$ has at most one root in $B$.
\end{proof}
\end{theorem}
For simplification, we set $\quad A=${\bf IsSMSys}$(F,X,B)$ as the algorithm to check whether $J_F^X(B)$ is an S-M system or not.  Interval evaluation can be used here to get $J_F^X(B)$.
%

Similar as Theorem \ref{preconditioner}, we have the following result.
\begin{theorem}\label{preconditioner2} Let $F=(f_1,\ldots,f_n), f_i\in C^1(B), 1\le i\le n$ and $\pp^*$ a simple zero of $F$. Then, $\exists \delta>0$ s.t. $\forall B\in B(\pp^*,\delta)$, $V J_F^{-1}(m(B))\cdot F^T$ is an S-M system in $B$, where $V$ is a S-M matrix.
\end{theorem}
Now, we will show how to check the existence for an S-M system inside a box. we give the following lemma first.
\begin{lemma}\label{recursion} Let $F=(f_1,\ldots,f_n), f_i\in C^1(B), 1\le i\le n$, $B=I_1\times \cdots \times I_n$, $F'=(f_1,\ldots,f_{n-1})$. If $F$ is an S-M system in $B$, then for any $\hat{x_i}\in I_i, i=1\ldots,n$, $F'|_{x_i=\hat{x_i}}$ is still an S-M system in $I_1\times \cdots \times I_{i-1} \times I_{i+1} \times \cdots \times I_n$.
\begin{proof} It can be derived from Definition \ref{SM} directly.
\end{proof}
\end{lemma}

Let $F$ be an S-M system in $B$ and $S=\V(F')$. By the definition of S-M system, we know that $S$ is strong monotone in $B$, meanwhile, we know that $F=0$ has at most one root in $B$. Thus, in Lemma \ref{case2}, if we replace the O-M system by S-M system, the lemma is still hold:
\begin{lemma}\label{SMcase2} Let $F=(f_1,\ldots,f_n)$ be an S-M system. Let $F'=(f_1,\ldots,f_{n-1})$ and $S=\V(F')$. Assume that $S\cap \partial B=\{\pp_1,\pp_2\}$, we have:
  \begin{enumerate}[(1)]
\item If $f_n(\pp_1)f_n(\pp_2)\leqslant 0$, $F=0$ has a unique root in $B$.
\item If $f_n(\pp_1)f_n(\pp_2)> 0$, $F=0$ has no root in $B$.
\end{enumerate}
\end{lemma}

Therefore, the existence condition for an O-M system can be applied for an S-M system. We still consider the index set:
$$Ind(B)=\{(i,a_i), (i,b_i) | i=1,\ldots,n\}.$$
For $(i, t)\in Ind(B)$, we also need to solve $F'|_{x_i=t}$ in $B_i$. Compared with the O-M system, the advantages of the S-M system are as follows:
\begin{enumerate}
  \item By Lemma \ref{recursion}, $F'|_{x_i=t}$ is still S-M system in $B_i$.
  \item Let $F'|_{x_i=t}=(f'_1,\ldots,f'_{n-1})$. If $f'_j(v(B_i))>0$ or $<0$ for some $j\in \{1,\ldots,n-1\}$, then by Lemma \ref{vertex}, we have $f'_j(B_i)>0$ or $<0$. Thus $F'|_{x_i=t}$ has no zero in $B_i$, we do not need to solve $F'|_{x_i=t}$.
\end{enumerate}

We give the following example to show how to check the existence.
\begin{example}\label{exam-existence} Let $F=(f_1,f_2,f_3)=(x-y+z,y^2+x+y+2z,x^2+yz-3x-y+z)$ and $B=[-0.1,0.1]\times[-0.1,0.1]\times[-0.1,0.1]$, $B_1,B_2,B_3$ are all $[-0.1,0.1]\times[-0.1,0.1]$. One can verify that $F$ is an S-M system in $B$. Next we check the existence consition. Let $F'=(f_1,f_2)$ and other notations are as above, we need to check whether these systems $F'|_{x_i=t_i}$ have a real zero in $B_i$ or not, where $x_1=x,x_2=y,x_3=z, t_i=-0.1,0.1$. (Some systems can be checked easily, for example, let $g=f_2|_{z=0.1}=y^2+x+y+0.2$, we have $g(0.1,0.1)>0$, $g(0.1,-0.1)>0$, $g(-0.1,0.1)>0$ and $g(-0.1,-0.1)>0$, i.e., $g(v(B_3))>0$. By Lemma \ref{vertex}, we know that $g(B_3)>0$, hence $F'|_{z=0.1}$ has no zero in $B_3$ and we do not need to solve the system). Finally we find both the two systems $F'|_{x=0.1}$ and $F'|_{x=-0.1}$ have a unique zero in $B_1$, see the left and right figures in Figure \ref{fig1}. Then we check the condition of Lemma \ref{SMcase2}, by simple evaluation, we have $f_3|_{x=0.1}(v(B_1))<0$ and $f_3|_{x=-0.1}(v(B_1))>0$. Therefore we know that $F=0$ has a unique zero in $B$.
\begin{figure}[ht]
\centering
\begin{minipage}{0.29\textwidth}
\centering
\includegraphics[scale=0.21]{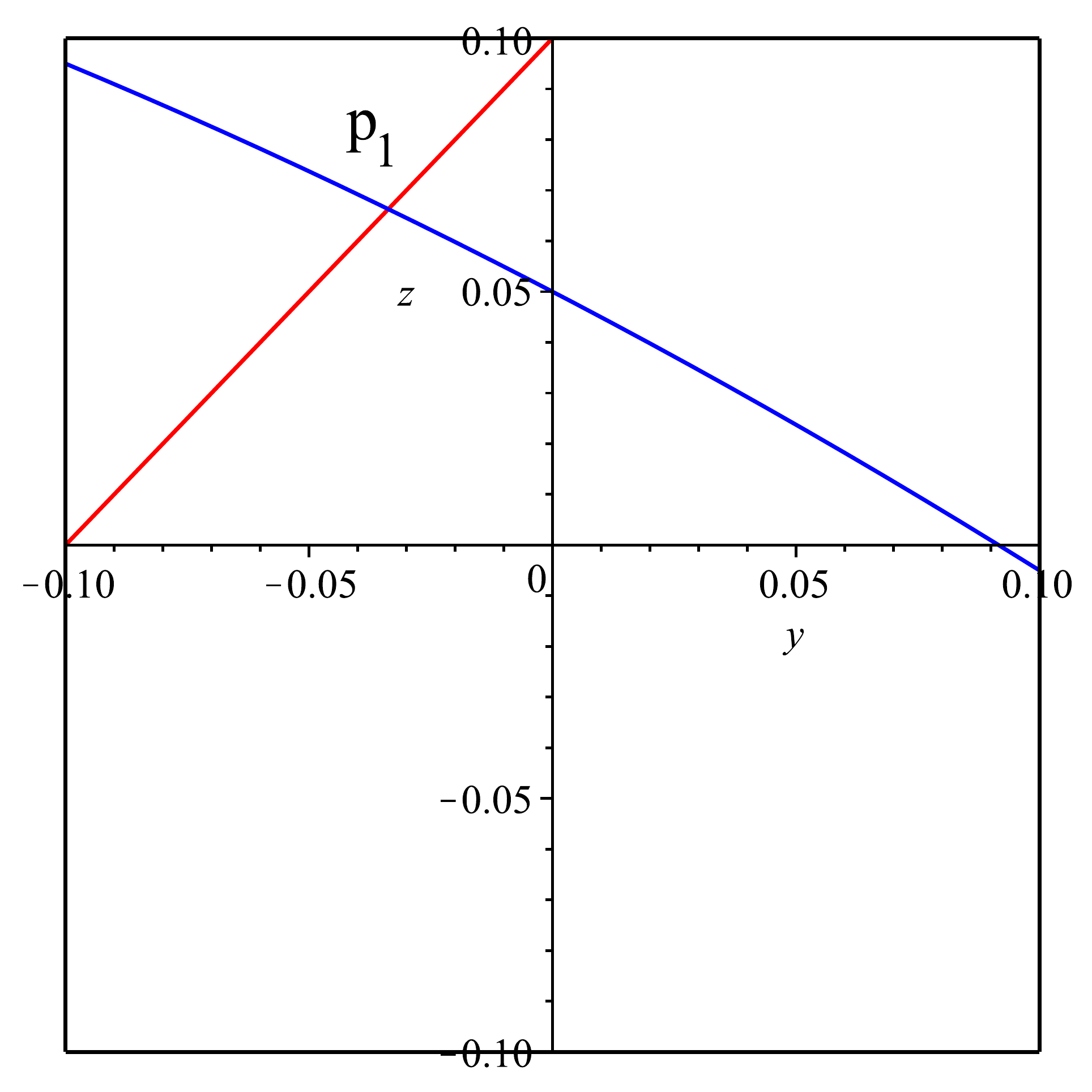}
\end{minipage}
\centering
\begin{minipage}{0.4\textwidth}
\centering
\includegraphics[scale=0.33]{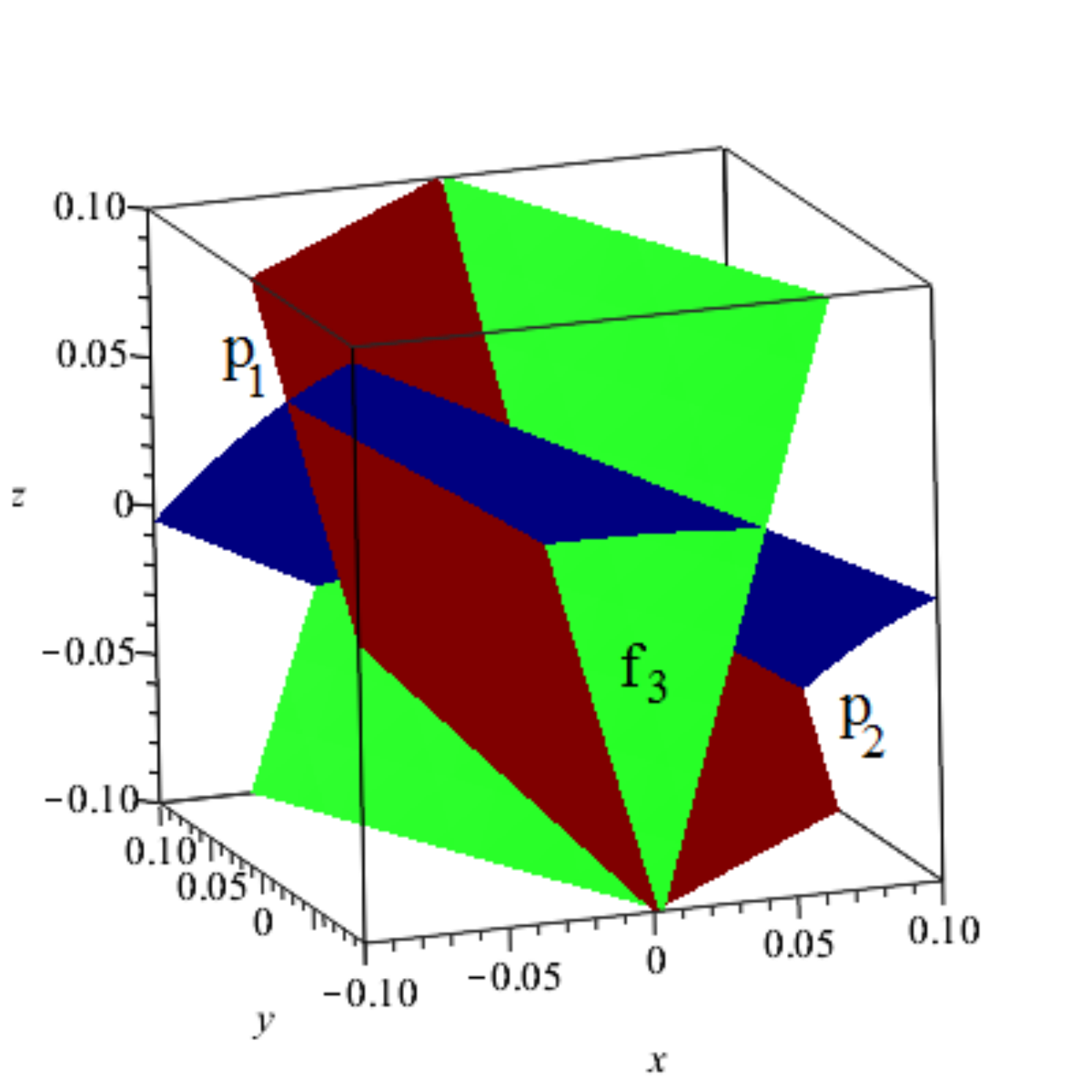}
\end{minipage}
\centering
\begin{minipage}{0.29\textwidth}
\centering
\includegraphics[scale=0.21]{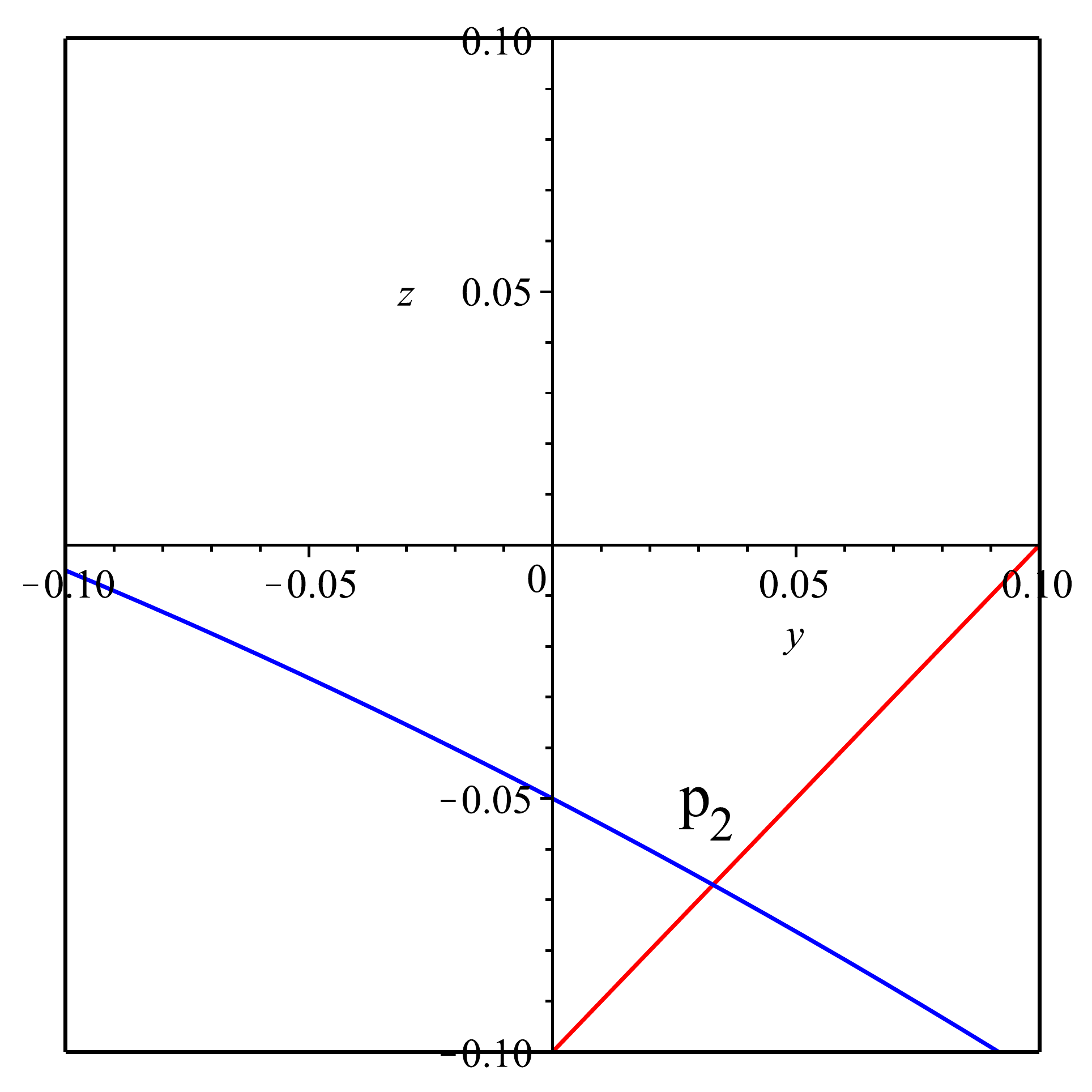}
\end{minipage}
\caption{The figures for Example \ref{exam-existence}}\label{fig1}
\end{figure}
\end{example}

Based on the discussion above, we give an algorithm {\bf Existence} below.
\begin{algorithm}[H] \label{alg-existence}
	\caption{$\quad A=${\bf Existence}$(F,B,\epsilon_b):$}
 \begin{enumerate}
\item $\textbf{Input}$: An S-M system $F=(f_1,\ldots,f_n)$, a box $B$, a given precision $\epsilon_b>0$ for refinement, where  $X=\{x_1,\ldots,x_n\}$ are variables of $F$.
\item $\textbf{Output}$: 1,0,or {\bf Unknown}.
\end{enumerate}

 \begin{enumerate}[1.]

     \item $F' \leftarrow (f_1,\ldots,f_{n-1}), T\leftarrow Ind(B), num\leftarrow0$, $S=\{\}$.
     \item While $T\neq \emptyset$ and $num <2$
     \begin{itemize}
         \item Take $(i,t)$ from $T$ and $T=T\setminus \{(i,t)\}$.
         \item $b\leftarrow{\bf Existence}(F'|_{x_i=t},B_i,\epsilon_b)$.
         \item If $b=${\bf Unknown}, return {\bf Unknown}.
         \item Else if $b=1$, $num=num+b$, $S\leftarrow S\bigcup \{(i,t)\}$.
     \end{itemize}
 \item If $num=0$, return 0.
 \item Else, $\{(i,t),(j,t')\}\leftarrow S$ 
        \begin{enumerate}[a.]
        \item $t_1\leftarrow  \mathrm{Sign}(f_n|_{x_i=t}(v(B_i))), t_2\leftarrow  \mathrm{Sign}(f_n|_{x_j=t'}(v(B_j)))$.
         \item While $w(B_i)>\epsilon_b$ and $t_1=0$ do
         \begin{itemize}
         \item  Refine $B_i$ w.r.t. $F'|_{x_i=t}$, and still denote the refined boxes as $B_i$.
         \item  $t_1\leftarrow  \mathrm{Sign}(f_n|_{x_i=t}(v(B_i)))$.
         \end{itemize}
         \item While $w(B_j)>\epsilon_b$ and $t_2=0$ do
         \begin{itemize}
         \item  Refine $B_j$ w.r.t. $F'|_{x_j=t'}$, and still denote the refined boxes as $B_j$.
         \item  $t_2\leftarrow  \mathrm{Sign}(f_n|_{x_j=t'}(v(B_j)))$.
         \end{itemize}
         \item If $t_1 t_2<0$, return 1.
         \item Else if $t_1 t_2>0$, return 0.
         \item Else, return {\bf Unknown}.
        \end{enumerate}
 \end{enumerate}
\end{algorithm}

\noindent{\bf Remarks for the algorithm}:
\begin{enumerate}
\item $A=1$ means that the system $F$ has a unique zero in $B$, $A=0$ means that the system $F$ has no zero in $B$.
\item When the algorithm returns Unknown, it means that we do not find two systems which have a unique zero in their corresponding boxes or can not determine the sign of $f_n|_{x_i=t}(v(B_i)) \cdot f_n|_{x_j=t'}(v(B_j))$. This case happens when the zero is on (or very close to) a $k$-dimensional ($k\le n-2$) the boundaries of $B$, so that our precision $\epsilon_b$ can not handle this case. This also explains why there are at most two $(n-1)$-dimensional faces of $B$ intersects $S$ in Step 4.  We can avoid this case by changing the length of the box or combining some adjacent boxes with the same output ``Unknown" to form a new box. Subdividing the new box and checking the conditions again, one usually succeeds in finding the results. The {\bf Existence} for bivariate systems are presented  in \cite{cheng5}.
\item In Steps 3.b, 3.c, refining the boxes $B_i, B_j$ means to refine the root of the systems $F'$ for $x_i=t, x_j=t'$ in the boxes to get smaller boxes. Notice that we know there exists and only exists one unique real zero in the related boxes. We can use interval-Newton method for the refinement. $B_i(B_j)$ or part of it can be set as the original box for the iteration.
\item The correctness and termination of the algorithm is based on theories before.
\end{enumerate}

\subsection{Choosing a proper S-M matrix $V$}
A better S-M matrix helps us reducing some unnecessary computation. In this subsection, we will discuss how to choose a ``better" S-M matrix $V$. As mentioned before, we may meet some bad cases: although both $F'|_{x_i=t}$ and $F'|_{x_j=t'}$ have a unique zero in $B_i$ and $B_j$, we can not determine the signs of $f_n|_{x_i=t}(v(B_i))$ and $f_n|_{x_j=t'}(v(B_j))$. Thus, we need to refine the boxes $B_i$ and $B_j$. Therefore, we want to select some ``nice" S-M matrices $V$ to avoid bad cases as much as possible. See the following example first.

\begin{example} Consider the example in \cite{cheng5}. Let $F=(y-x^2,x-2y)$, $B=[-0.1,0.1]\times[-0.1,0.1]$.
 Let $V_1=\begin{pmatrix} 1 & 1 \\ 1 & -1 \end{pmatrix}$ and $V_2=\begin{pmatrix} 2 & 1 \\ 1 & -2 \end{pmatrix}$. We have $G_1=V_1 J_F^{-1}(m(B))\cdot F^T=(-3x^2+x+y,-x^2+x-y)^T=(g_1^{(1)},g_2^{(1)})^T$ and $G_2=V_2 J_F^{-1}(m(B))\cdot F^T=(-5x^2+2x+y,x-2y)^T=(g_1^{(2)},g_2^{(2)})^T$. It is easy to check that $G_1$ and $G_2$ are both S-M systems in $B$. Next we consider the existence condition. For the system $G_1$, $g_1^{(1)}|_{x=0.1}$ has a unique zero in $B_1=[-0.1,0.1]$, however, $g_2^{(1)}|_{x=0.1}(v(B_1))$ are not all positive(negative), then we need to refine $B_1$, see the left figure in Figure \ref{fig_chooseU}.  For the system $G_2$, both $g_1^{(2)}|_{y=0.1}$ and $g_1^{(2)}|_{y=-0.1}$ have a unique zero in $B_2=[-0.1,0.1]$, then we can get $g_2^{(2)}|_{y=0.1}(v(B_2)<0$ and $g_2^{(2)}|_{y=-0.1}(v(B_2))>0$ immediately, thus we know the existence without refinement, see the right figure in Figure \ref{fig_chooseU}. It means that $V_2$ is ``better" than $V_1$.
\begin{figure}[htbp]
\centering
\begin{minipage}{0.3\textwidth}
\centering
\includegraphics[scale=0.2]{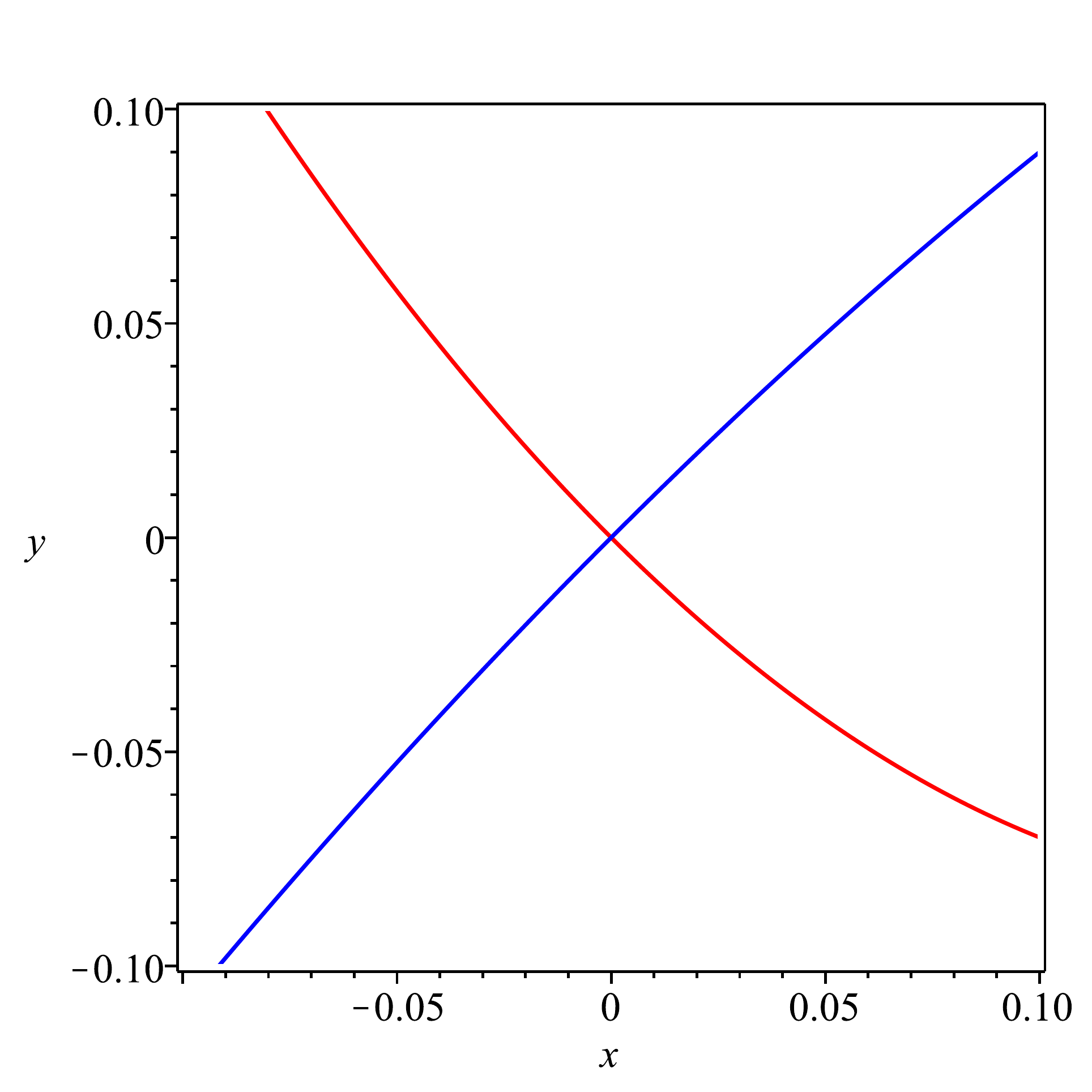}
\end{minipage}
\centering
\begin{minipage}{0.3\textwidth}
\centering
\includegraphics[scale=0.2]{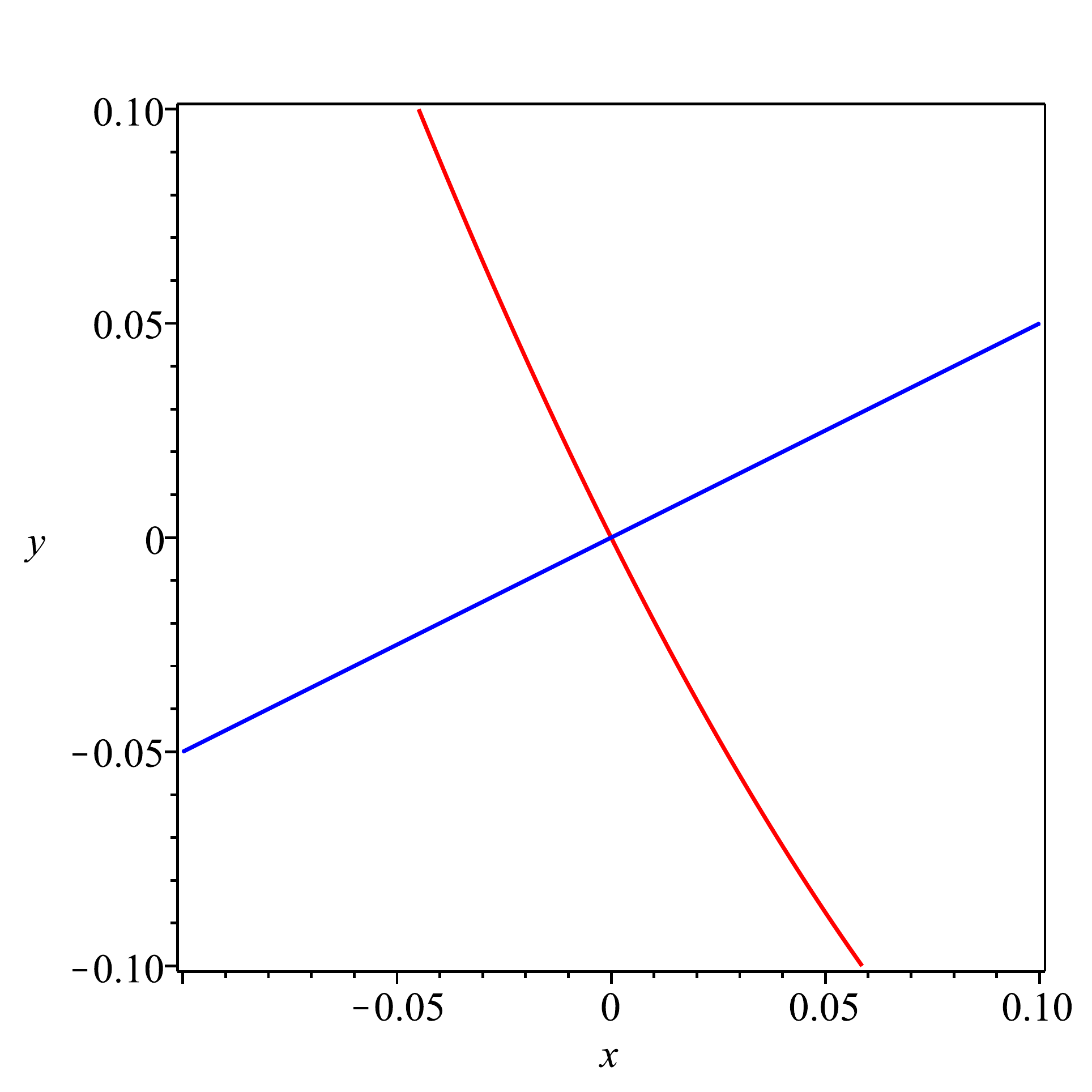}
\end{minipage}
\caption{The influlence of the different S-M matrixes to the same system.
} \label{fig_chooseU}
\end{figure}
\end{example}

Give a system $F=(f_1,\ldots,f_n)$ and a box $B$, we assume that $B$ is an isolating box of $F$ (i.e., $B$ contains only one simple zero $\pp$ of $F$). Let $G=V J_F^{-1}(m(B))\cdot F^T=(g_1,\ldots,g_n)^T$ and $G'=(g_1,\ldots,g_{n-1})$. Our goal is to choose an S-M matrix $V$ such that $S=\V(G')$ goes through $F_n^l(B)$ and $F_n^r(B)$, i.e., both $G'|_{x_n=a_n}$ and $G'|_{x_n=b_n}$ have a unique zero in $B_n$, meanwhile, $g_n|_{x_n=a_n}(v(B_n))\cdot g_n|_{x_n=b_n}(v(B_n))<0$. Then we immediately know $B$ is an isolating box of $F$. Notice that when $w(B)$ is small, i.e., $m(B)$ is close to the simple zero $\pp=(p_1,p_2,\dots,p_n)$, we know that
$J_F^{-1}(m(B))\cdot F^T\approx(x_1-p_1,\ldots,x_n-p_n)^T$ and
$G\approx V\cdot (x_1-p_1,\ldots,x_n-p_n)^T$ in $B$.
That is to say $g_i$ is approximately a hyperplane in $B$ and $\nabla g_i(m(B)) \approx V_{i,:}$ for $i=1,\ldots,n$. 

Based on the analysis above, we assume that $\pp=\bf{0}$, $F=(x_1+h_1,\ldots,x_n+h_n)$ and $B=[-1,1]^n$ is a unit box, where $h_i\in C^1(B)$ and the Taylor expansion of $h_i$ at $\pp=\bf{0}$ has only terms with degree greater than 1, $i=1,\ldots,n$. Then we show how to choose $V$ such that the system $G=(g_1,\ldots,g_n)^T=V\cdot F^T$ satisfying the following conditions:
\begin{enumerate}[(1)]
\item $V$ is an S-M matrix.
\item $(g_1,\ldots,g_{n-1})|_{x_n=a_n}$ and $(g_1,\ldots,g_{n-1})|_{x_n=b_n}$ have a unique zero in $F_n^l(B)$ and $F_n^r(B)$.
\item $  \mathrm{Sign}(g_n(F_n^l(B))) \cdot   \mathrm{Sign}(g_n(F_n^r(B)))<0$.
\end{enumerate}

For example, $V_{2d}=\begin{pmatrix} N & 1 \\ 1 & -N\end{pmatrix}$ for $N=2$ or a larger positive integer and $V_{3d}=\begin{pmatrix} N & 1 & 1 \\ 1 & -N & 1 \\ 1 & 1 & N \end{pmatrix}$ for $N=3$ or a larger positive integer satisfy our conditions, see Figure \ref{chooseU2}.

\begin{figure}[htbp]
\centering
\begin{minipage}{0.4\textwidth}
\centering
\includegraphics[scale=0.22]{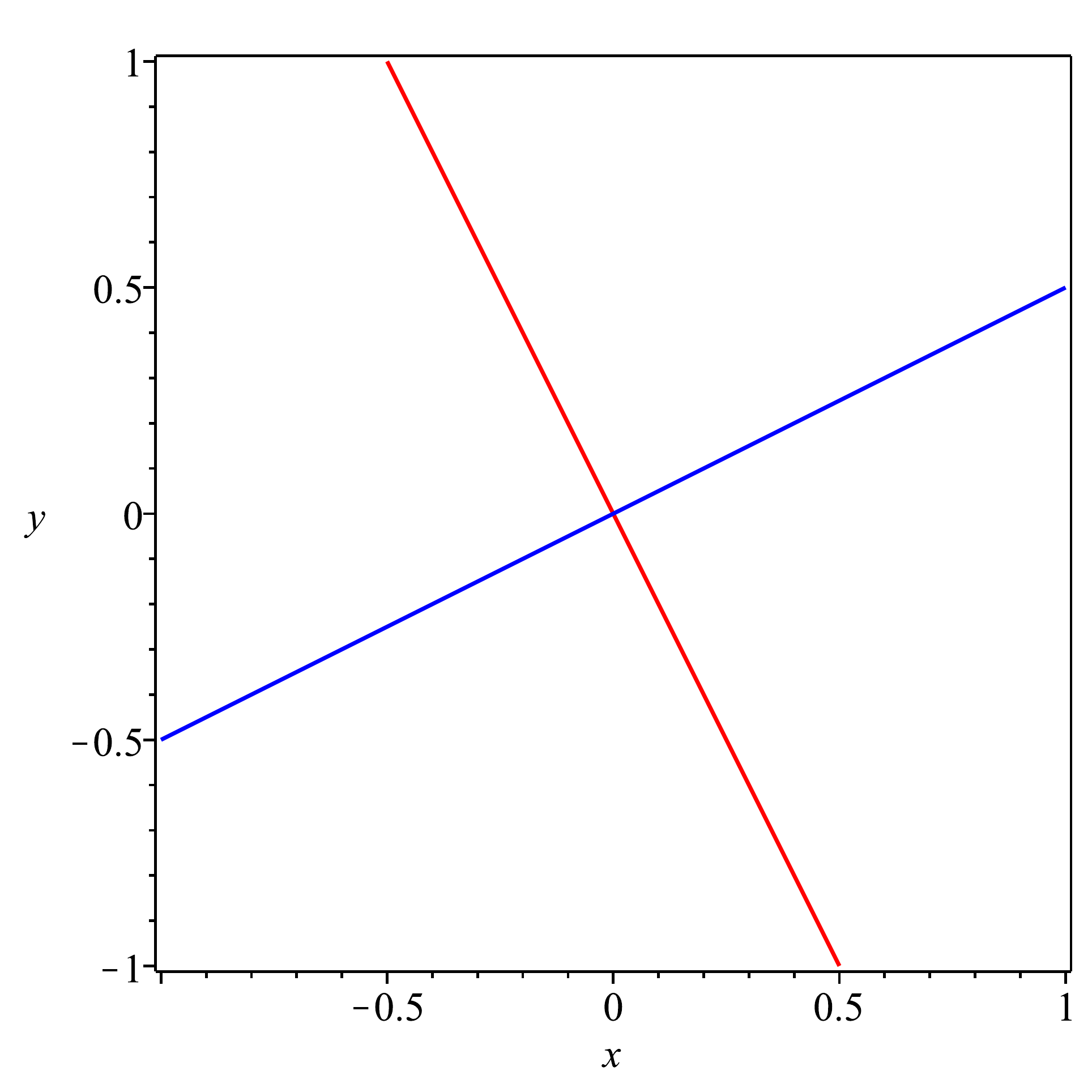}
\centerline{$V_{2d}\cdot (x,y)^T$}
\end{minipage}
\centering
\begin{minipage}{0.4\textwidth}
\centering
\includegraphics[scale=0.33]{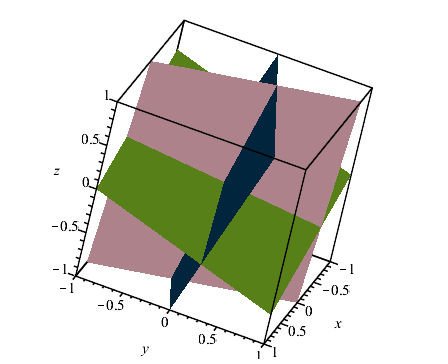}
\centerline{$V_{3d}\cdot (x,y,z)^T$}
\end{minipage}
\caption{Examples for choosing $V$.} \label{chooseU2}
\end{figure}
For general cases,
\[V_{nd}=\begin{pmatrix} \alpha_{11} & \alpha_{12} & \cdots & \alpha_{1n} \\ \alpha_{21} & \alpha_{22} & \cdots & \alpha_{2n} \\ \vdots & \vdots  & \vdots & \vdots \\ \alpha_{n1} & \alpha_{n2} & \cdots & \alpha_{nn} \end{pmatrix}.\]
Let $\alpha_{ii}=(-1)^{i+1}N$ and $|\alpha_{ij}|\le 1, i\ne j$, by generating $\alpha_{ij}$ randomly, we can find an S-M matrix $V_{nd}$ satisfies the above conditions with high probability, where $N$ can be $n$ or another integer larger than $n$. We can also choose $\alpha_{ii}$ for different sign(s) for more choices.  Usually, it is easier to get an S-M matrix for $\frac{\partial G}{\partial X}(B)$ when $N$ is larger. This phenomenon can be obtained by observing the matrix of $\frac{\partial G}{\partial X}(B)$.
We can also check whether it is a matrix we want by Definition \ref{OMdef} and we ensure that it is an S-M matrix.

\subsection{Algorithm}
In this subsection, we give the main algorithm for real zero isolation of real nonlinear systems. Our method is a subdivision method, we need the exclusion test which is based on the following famous box predicate $C_0(f,B)$\cite{PV}:
\begin{equation}\label{eqn-exclusion}
C_0(f,B):= 0\notin \Box f(B).
\end{equation}
Given a function $f$ and a box $B$, we say $C_0(f,B)$ is true if $0\notin \Box f(B)$. Based on the box predicate, we can write it as an algorithm
$A=${\bf Exclusion}$(F,B)$.
%
%
Obviously, if the algorithm {\bf Exclusion} returns 1, $F=0$ must have no root in $B$.

If the given system is a polynomial system, we can isolate the real zeros of bounding polynomials to exclude boxes \cite{cheng5}.
We rewrite a multivariate polynomial as below.
$$f_i(x_1,\ldots,x_n)=\sum_{j=0}^{d_i} t_{i,j}(x_1,\ldots,x_{n-1}) x_n^j.$$
Let $B=I_1\times\cdots\times I_n=B_1\times I_n$. We split $B_1$ into small boxes with a given length. For each these small $(n-1)$-D box $\mathbf{b}$, we evaluate $f_i(x_1,\ldots,x_n)$ on $\mathbf{b}$ to get a sleeve polynomial \cite{Cheng1}.
$$f_i(\mathbf{b},x_n)=\sum_{j=0}^{d_i} t_{i,j}(\mathbf{b}) x_n^j=\sum_{j=0}^{d_i} [a_{i,j},b_{i,j}] x_n^j.$$
Isolating the real zeros of $f_1(\mathbf{b},x_n)$ in $I_n$ (see \cite{cheng5} for details), we can get a list of intervals, say $J_1,\ldots,J_m$. Continuing to isolate the real zeros of $f_2(\mathbf{b},x_n)$ in $J_1,\ldots,J_m$, we can get a list of intervals, say $J'_1,\ldots,J'_{m'}$, or an empty set. If we get an empty set, the box $\mathbf{b}\times I_n$ can be thrown away. Else, doing so for $f_3,\ldots,f_n$, we can get a list of intervals, say $\bar{J}_1,\ldots,\bar{J}_{m''}$ or an empty set. Then we can get candidate boxes $\mathbf{b}\times \bar{J}_k (1\le k\le m'')$ or throwing away the box $\mathbf{b}\times I_n$. Thus we exclude some sub-boxes of $B$. 
We can recursively do so on the boxes to exclude some sub-domain of $B$. When $I_n$ is large, this method may be more efficient to compute the possible candidate boxes than the method in (\ref{eqn-exclusion}).
In fact, this method also works for non-polynomial systems. But the way to construct the upper (lower) bounding function is a little more complicated than polynomial case and we need to isolate the real zeros of non-polynomial univariate equations. We denote it as an algorithm $A=${\bf Candidate}$(F,B)$. Our implementation uses mainly this method for polynomial systems.

Based on the discussion above, we have the main algorithm {\bf Realrootfinding}.

\begin{algorithm}[H] \label{algo-main}
    \caption{$\quad \mathbf{R},\mathbf{SR}=${\bf Realrootfinding}$(F,B_0,\epsilon):$}
\begin{enumerate}
\item $\textbf{Input}$: A system $F=(f_1,\ldots,f_n)$, an initial box $B_0$, a termination precision $\epsilon>0$.
\item $\textbf{Output}$: An isolating box set $\mathbf{R}$ and a suspected root box set $\mathbf{SR}$.
\end{enumerate}
 \begin{enumerate}[1.]
 \item $\mathbf{R}\leftarrow\emptyset, \mathbf{SR}\leftarrow \emptyset$ and $\mathbf{BS} \leftarrow\{B_0\}$.
 \item While $\mathbf{BS}\ne \emptyset$
       \begin{enumerate}[(1)]
              \item Take $B\in \mathbf{BS}$, $\mathbf{BS}\leftarrow \mathbf{BS}\setminus\{B\}$, $A \leftarrow ${\bf Exclusion}$(F,B).$
              \item If $A=0$,
              \begin{enumerate}[(a)]
                    \item $G\leftarrow V_{nd} J_F^{-1}(m(B))\cdot F^T.$
                    \item $A'\leftarrow${\bf IsSMSys}$(G,B).$
                     \item If $A'=1$,  
                                   \begin{itemize}
                                   \item $A''\leftarrow${\bf Existence}$(G,B,\epsilon)$.
                                   \item If $A''=1$,  $\mathbf{R}\leftarrow \mathbf{R}\bigcup\{B\}$.
                                   \item If $A''=${\bf Unknown}, $\mathbf{SR}\leftarrow \mathbf{SR}\bigcup\{B\}$.
                                   \end{itemize}
                     \item Else if $w(B)>\epsilon$, 
                                   \begin{itemize}
                                   \item Split $B$ into two similar parts and add them into $\mathbf{BS}$.
                                   \end{itemize}
                     \item Else, $\mathbf{SR}\leftarrow \mathbf{SR}\bigcup\{B\}$.

              \end{enumerate}
       \end{enumerate}
 \item Return $\mathbf{R}$,$\mathbf{SR}$.
 \end{enumerate}
\end{algorithm}

The correctness of the algorithm {\bf Realrootfinding} is guaranteed by Theorem \ref{SM} and Lemma \ref{SMcase2}. The termination of the algorithm is guaranteed by Theorem \ref{preconditioner2} and the given $\epsilon$.

\noindent{\bf Remarks for the Realrootfinding algorithm}:

1. For a system $F$ with only simple zeros, we can always get all the isolating boxes of all the zeros of $F$ in $B_0$ by recursively subdividing those suspected root boxes. Our theories ensure the termination of the algorithm.

2. If $F$ has singular zeros, we can not determine whether a box contains only a singular zero or not. Thus, we give the termination precision $\epsilon>0$ and repeat subdividing the boxes until the widths of the obtained boxes are less than $\epsilon$. Finally we get some suspected root boxes. Each suspected root box may contain several zeros (counting the multiplicities of the zeros) or no zero.

3. In order to get all the real zeros of a given system in a general given real box, we consider the coordinate transformation: $x_i\to \frac{1}{x_i}$. We map the interval $[-b,-1]$ to $[-1,-1/b]$, and the interval $[1,b]$ to $[1/b,1]$, where $b>1$. Hence, we need only to consider finding real zeros in $[-1,1]^n$. Doing this way, we take only interval computation inside $[-1,1]^n$ which need less interval evaluation. If the given system has only finite real zeros in a bounded box, we can get all its real zeros in the whole real space. For example, for a bivariate system $F=(f_1(x,y),f_2(x,y))$, we can get the isolating boxes or suspected boxes of real zeros of the original system in $[1,b]\times [-1,1]$ by isolating the real zeros of the system $(f_1(\frac{1}{x},y), f_2(\frac{1}{x},y))$ in $[1/b,1]\times[-1,1]$. If $F$ is a polynomial system, then we can get all the real zeros of $F$ if we take $b$ as its root bound.

4. In Step 2.(2) (d), we can find that some regions in $B$ may be computed for several times which waste much computing time. Thus we can split the given box $B_0$ into many smaller boxes at first.

5. In Step 2.(1), our aim is to find the candidate regions which may contain real zeros of the given system. For the case that the give system is a polynomial system, we can compute candidate regions as below (see more details in \cite{Cheng1,Cheng3,cheng5}).
Given a system $F=(f_1,\ldots,f_n)$ and a box $B=I_1\times \cdots \times I_n$, let $J=I_2\times \cdots \times I_n\subset \R^{n-1}$. We denote $T_i$ as the set of the real (interval) zeros of the interval polynomial $f_i(x_1,J)$ for $i=1,\ldots,n$, and let $\{t_1 \cap \cdots \cap t_n \cap I_1| t_i\in T_i,i=1,\ldots,n\}=\{I'_1,\ldots,I'_m\}$. We call $\{I'_1\times J,\ldots,I'_m\times J\}$ the {\bf candidates} of $F$ in $B$. It is obvious that all the real zeros of $F$ in $B$ are in the candidates. If the width $w(J)$ is large, we can split $J$ into $(n-1)^k$ parts equally (by splitting $I_i, i=2,\ldots,n$ into $k$ parts equally) and compute the candidates separately.

6.
For each suspect box we got,  it may contain no zero, a simple zero, a multiple zero, several simple zeros or, one or more simple zero(s) together with one or more multiple zero(s) of the given system. This happens because we set a termination precision for the subdivision precess since we are not sure if the given system has multiple real zeros or not. For some examples, there may be so many suspect boxes and most of them contain no roots. We need to remove the redundant boxes which contains no roots. For a given system, there may exist one or several cluster(s) of boxes. Each cluster of boxes may contain one (or several) multiple (or simple) root(s), or no roots. If there is a root, the Newton's method will converge to the root if the start point is chosen well, that is, the start point is in the basins of attraction of the system for the root \cite{SUSANTO20091084}. The convergent region for the root will intersect some the boxes inside the cluster of the boxes. If we choose properly some point(s) in each suspect box in the cluster of boxes as start point(s) for Newton's method for the system, we may get the root(s) inside the cluster of boxes and remove the redundant boxes. For the derived box(es) after computing with Newton's method, we can do only a heuristic verification of a suspect box by deflation methods (see \cite{deflation,multiplicity,Leykin2006NewtonsMW} and the methods mentioned therein). Notice that we may miss some root(s) or get more roots with this operation. For example, a root on (or very close to) the boundaries of a suspected box, the root may be missed or counted twice because of numerical computation. But it usually works well. We will show experiments for illustration with this step.

%
%

\subsection{Complexity analysis}
We analyze the bit complexity of isolating the real roots of a zero-dimensional polynomial system $\Sigma$ in this subsection.  We assume that there are $n$ variables and the degrees of the polynomials of the system are bounded by $d$, the number of their terms are bounded by $m$ and the bit sizes of their coefficients are bounded by $\tau$.

The complexity analysis for subdivision based algorithm of a single polynomial was considered by \cite{PVburr,PVcucker}. The complexity analysis for subdivision based algorithm for a polynomial system was given in \cite{Mantzaflaris}. Different from their work to find exact results, our analysis is for a given terminating precision without assuming that the system has only simple real roots.

We would like to mention that the condition number is an important parameter for the complexity analysis of subdivision based methods for root finding. There are a series of works about real root counting of polynomial systems with probatilistic numeric methods and related analysis based on condition number \cite{CUCKERrootcount1,CUCKERrootcount2,CUCKERrootcount3}.

In this paper, $\mathcal{O}$ means the bit complexity, $\O(\cdot)$ indicates that we omit poly-logarithmic factors.

Consider the real roots of $\Sigma$ inside $B=[-1,1]^n$. Notice that if we want to get all the real zeros of $\Sigma$, we can transform the original system into new system by replacing $x_i$ with $1/x_i$ and removing the denominators of the whole polynomials. We can get $2^n$ this kind of systems in $[-1,1]^n$. We assume that the termination precision for the boxes is $\rho $, that is, we stop subdividing the boxes when their lengths are less than $\rho$, where $\rho$ is a rational number such that $0<\rho<1$. So the number of the boxes in $B$ is bounded by $(\frac{2}{\rho})^n$. The bitsize of the endpoints of the boxes is bounded by $-log(\rho)$.

For each box, we take one exclusion test and one our existence test at most. We will analyse these two operations one by one.

\begin{lemma}[Multivariate~polynomial~interval~evaluation]\label{multivar-interval-evaluation}
  Let $g\in\Z[x_1,\ldots,x_n]$ be of magnitude $(d,\,\tau)$ and $m$ terms.
  $I_1,\ldots,I_n$ are intervals whose endpoints are rational numbers with bitsize
  $\sigma$, then evaluating
  $g(I_1,\ldots,I_n)$ has a bit complexity of
  $\O(d\,m\sigma+m\tau)$, and the bitsize of the endpoints of $g(I_1,\ldots,I_n)$ is $\mathcal
  {O}(d\sigma+\tau)$.
\end{lemma}
\begin{proof}
Consider one term of $g$ at first.  The operations here include the multiplications of $d$ intervals at most and one multiplication between the coefficient and the product of the intervals. We can divide $d$ intervals into $\frac{d}{2}$ pairs. Each pair usually contains two intervals. If $d$ is odd, the last pair contains only one interval. To get the product of them includes at most 4 multiplications. The total bit complexity for computing all the pairs is $\frac{d}{2}*4*\sigma=2\,d\sigma$. The bitsizes of the products are $2\sigma$. For these $\frac{d}{2}$ products, we divide them into $\frac{d}{4}$ pairs. Similarly, the total bit complexity for computing the products of all the pairs is $\frac{d}{4}*4*2\sigma=2\,d\sigma$. The bitsizes of the products are $4\sigma$. Doing  so in a similar way until the $k$ step such that $d=2^k$, that is, $k=\log(d)$, we can get the product of all the intervals.
So to get the product of all the intervals, we have the total bit complexity $2\,d\sigma*\log(d)$ and the bitsize of the final product is $2^{\log(d)}\sigma=d\,\sigma$. Considering the coefficient into the product, we have the total bit complexity of evaluating one term is $2\,d\sigma*\log(d)+d\,\sigma+\tau=\O(d\sigma+\tau)$.  The bitsize of the final product for one term is $d\sigma+\tau$.  So evaluating
  $g(I_1,\ldots,I_n)$ has a complexity of
  $\O(d\,m\sigma+m\tau)$, and the bitsize of the endpoints of $g(I_1,\ldots,I_n)$ is $\mathcal
  {O}(d\sigma+\tau)$.
\end{proof}


\begin{lemma}\cite{matrixinverse} \label{lem-minverse} Let $A=(a_{i,j})\in Z^{n\times n}$  be nonsingular. We denote by $\| A \| := \max |a_{i,j}|$ the maximum magnitude of
entries in $A$, and by $\kappa(A) := \|A\| \|A^{-1}\|$ the condition number of the matrix with respect
to the max norm. We give a Las Vegas probabilistic algorithm that has expected running
time $\O(n^3 (log \|A\| + log \kappa(A) ))$ bit operations to compute the exact inverse of $A$. Thus,
for a well conditioned $A$, with $\kappa(A)$ bounded by a polynomial function of $n log\|A\|$, this
cost estimate becomes $\O(n^3 log \|A\|)$.
\end{lemma}

\begin{lemma}\label{lem-nlcm}
We can compute the least common multiple of $n$ integers with bitsize $\sigma$ by $\O(n\sigma) $ bit complexity and the bit size of the least common multiple is $n\sigma$.
\end{lemma}
\begin{proof}
 We can divide $n$ intervals into $\frac{n}{2}$ pairs. Each pair is two integers. Note that if $n$ is odd, the last one can be regarded as a pair. For each pair, we compute the least common multiple of the two integers $a,b$. It is $a\, b/\gcd(a,b)$. The bit complexity is $2\O(\sigma)$. So the bit complexity of computing all the pairs is $n\O(\sigma)$.  The bit sizes of the results are all $2\sigma$.  We continue to divide the results into $\frac{n}{4}$ pairs. For each pair we compute its least common multiple with bit complexity $4\O(\sigma)$. So the bit complexity of computing all the pairs is $n\O(\sigma)$. And the bit sizes of the results are $ 4\sigma$.  Doing so, until we get the least common multiple of the $n$ integers, which we need $log(n)$ steps.  So the total bit complexity is  $log(n)\O(n\sigma)=\O(n\sigma)$. The bit size of the last least common multiple is $2^{log(n)}\sigma=n \sigma$.
\end{proof}

For a nonsingular matrix $M$ with rational entries such that the bit sizes of the entries are bounded by $\sigma$, we can rewrite $M=M_1 M_2$, where $M_1$ is a diagonal matrix and $M_2$ is a matrix with integer entries. We can compute the least common multiple of the denominators of the elements of each row. Set its inverse as the element of the related row of $M_1$. Each of the related element of $M$ of the row multiplies the least common multiple and set them as the related elements of $M_2$. The bit size of each element of $M_2$ is $n\sigma$, so as $M_1$. By Lemmas \ref{lem-minverse}, \ref{lem-nlcm}, we can find that the bit complexity of computing the inverse of $M$ is $\O(n^4\sigma+n^3\,log (\kappa(\|M_2\|))$ ( a well conditioned one is $\O(n^4\sigma)$)  and the bit sizes of the elements of $M^{-1}=M_2^{-1}M_1^{-1}$ is $2n\sigma$.

Now we consider the complexity to check the existence and the uniqueness of a real root inside a box.  When we compute the intersection between the space curve formed by the $n-1$ functions and the boundaries of a box. We need to check each face of the box to intersect the space curve. There are $2\,n$ faces. Each face is related to a zero-dimensional system with $n-1$ functions and $n-1$ variables. Recursively, we will do root finding of $2^{n-1}\,n!$ univariate polynomials. In order to get an approximating root, we can bisect the interval a fixed number times, say 10 times, if there is a root. We also need to multiply two square matrices with order $n$: One is $V$, the other is $J_F^{-1}(m(B))$.
\begin{lemma}\label{lem-existence}
  Let $F=(f_1,\ldots,f_n)\subset\Z[x_1,\ldots,x_n]$ and each $f_i$ be of magnitude $(d,\,\tau)$, and $m$ terms. 
  If $B=I_1\times\cdots I_n$ is a box with rational intervals such that the endpoints of $I_1,\ldots,I_m$ all with bitsize
  $\sigma$, then checking whether
  $F$ contains a unique real root in $B$ has a bit complexity of
  $\mathcal{O}(2^n n^{n+1}\,m\,(d\sigma+\tau))$.
\end{lemma}
\begin{proof} Denote the Jacobian matrix of $F$ w.r.t. $x_1,\ldots,x_n$ as $J_F$. The bitsize of the middle point of $B$, say $\P$, is $\sigma$. It is clear that the bitsize of each element of $J_F(\P)$ is $\mathcal{O}(d\sigma+\tau)$. The bit complexity to compute $J_F(\P)$ is $\mathcal{O}(n^2\,m\,d\sigma+n^2\,m\,\tau)$ by Lemma \ref{multivar-interval-evaluation}. Thus the bitsize of the elements of the inverse of $J_F^{-1}(\P)$ is $\mathcal{O}(n\,d\sigma+n\,\tau)$ and the bit complexity to compute $J_F^{-1}(\P)$ is $\mathcal{O}(n^5\sigma+n^4\tau)$ by the analysis below Lemma \ref{lem-nlcm}. The bit complexity to compute $V\,J_F^{-1}(\P)$ is $n^3\mathcal{O}(n\,d\sigma+n\,\tau)=\mathcal{O}(n^4\,d\sigma+n^4\,\tau)$, where $V$ is the S-M matrix we mentioned before whose elements are with bitsize $\mathcal{O}(1)$. And the bitsize of the elements of $V\,J_F^{-1}(\P)$  is $\mathcal{O}(n\,d\sigma+n\,\tau)$.

Let $G=(g_1,\ldots,g_n)=V\,J_F^{-1}(\P) F^T$. Then $J_G=V\,J_F^{-1}(\P) J_F$.
After we evaluate $J_F$ on $B$, denoted as $J_F(B)$, we have that each element of the matrix has a bitsize of  $\mathcal
  {O}(d\sigma+\tau)$ by Lemma \ref{multivar-interval-evaluation}. Notice that $f_i$ and $\frac{\partial f_i}{\partial x_j}$ have the same bitsize after evaluating on $B$.  In order to check whether the matrix $J_G(B)=V\,J_F^{-1}(\P) J_F(B)$ is strong monotonous over $B$, we compute the minors of the determinant of $J_G(B)$ step by step.  We compute the order $i+1$ minors with the result of the order $i$ minors until we get the determinant of $J_G(B)$, where changes from 1 to $n-1$. Totally, we can consider computing $n!$ products among $n$ intervals. From the way we check the strong monotonous matrix condition, we compute each product by multiplying the $n$ intervals one by one.  There are $4(n-1)$ multiplications. Notice that the bit size of $J_G(B)$ is $\mathcal{O}(n\,d\sigma+n\,\tau)$. The total complexity to get one product is
  {\tiny $$4*\mathcal
  {O}(n\,d\sigma+n\,\tau)+4*2\mathcal
  {O}(n\,d\sigma+n\,\tau)+\ldots+4*(n-1)\mathcal
  {O}(n\,d\sigma+n\,\tau)=2*n*(n-1)\mathcal
  {O}(n\,d\sigma+n\,\tau)=\mathcal
  {O}(n^3\,d\sigma+n^3\tau).$$}
  By Stirling's approximation, the total bit complexity to check the strong monotonous condition is $\mathcal{O}(n^n) \mathcal{O}(n^3\,d\sigma+n^3\tau)=\mathcal
  {O}(n^{n+3}\,d\sigma+n^{n+3}\tau)$. Notice that the uniqueness condition is already checked by the determinant of $J_G(B)$. To check the number of intersections between the space curve formed by $g_1,\ldots,g_{n-1}$ and the faces of $B$, we need to check whether $2^n\,n! $ univariate polynomials with magnitude $(d,n\,d\sigma+n\,\tau)$ and at most $n\,m$ terms have solutions in the related intervals.
  So for one root isolation on one interval, the bit complexity is $\mathcal{O}(n\,m(d\sigma+n\,d \sigma+n\tau))=\mathcal{O}(n\,m\,d\sigma+n\,m\tau)$ by Lemma \ref{multivar-interval-evaluation}.
 So the total bit complexity for the root isolation is $2^n n! \mathcal{O}(n\,m\,d\sigma+n\,m\tau)=\mathcal{O}(2^n n^{n+1}\,m\,(d\sigma+\tau))$. It is not difficult to find that this is the main part for the whole bit complexity when compared all the steps. Thus the total bit complexity of checking whether
  $F$ contains a unique real root in $B$ is $\mathcal{O}(2^n n^{n+1}\,m\,(d\sigma+\tau))$.
\end{proof}

We want to mention that there exist more efficient algorithm to compute the determinant of an interval matrix, see \cite{det_interval_matrix}. But it does not change the total bit complexity. The analysis is based on our implementation.

From the analysis above, we can directly deduce the following result.
\begin{theorem}\label{thm-complexity}
Isolating the real roots of a square system of polynomials with magnitude $(d,\tau)$ and $m$ terms with the algorithm mentioned above, one takes a bit complexity of
$\mathcal{O}((\frac{8}{\rho})^n n^{n+1}\,m\,(-d\, log(\rho)+\tau))$.
\end{theorem}


\section{Experiments}
We implement our algorithm in C++, and we do some experiments under Linux with a computer of 64 AMD 3990X 2.90GHz CPU and 64 GB RAM. We also realize the parallel computation with MPI. Our code can be downloaded from http://mmrc.iss.ac.cn/ \~ \,jcheng/pai/RootFinding.tar.gz as well as a simple user guider from http://mmrc.iss.ac.cn/ \~ \,jcheng/pai/Pai User Guide.pdf.
We compare our algorithm with MK test \cite{GS01, Kioustelidis,Moore2} and Bertini \cite{bertini}. There are two other famous homotopy continuation softwares: PHCpack \cite{PHCpack} and Hom4ps \cite{Hom4ps-3}. We choose Bertini to compare just because it derives real roots directly.

For systems with small sizes, we check their roots with symbolic methods such as \cite{Cheng3} and we will point out if the results of some methods are not correct.

We use NiDj to represent the systems formed by polynomials with $i$ variables and degree $j$ in the tables. NiDjE means we scaling back the coordinates of the variables. For example, the polynomials in N2D5E case is got by substituting variables $(x_1,x_2)$ with $(100*x_1,100*x_2)$, and then we can just compute the roots in $[-1,1]^2$ to get the roots in $[-100,100]^2$ of the original system, which is more efficient in some cases. In the tables, Case means the type of systems. Terms means the number of terms of each polynomial in the system. Coef means the maximal absolute value of the coefficients of the polynomials in the systems. Range means the box to search the real roots for our method and MK test. But Bertini finds all the complex roots of the given systems, including the real roots out of the given box. Roots means the number of real roots of the systems. For our method, there are three parameters: The first one is the number of certified real roots of the system inside the box; The second one is the number of the suspected boxes of the system inside the box; The last one is the number of possible real roots of the system in all the suspected boxes with the method mentioned in Remark 6 of Algorithm 2.
MK contains only the first two parameters for their roots. Bertini shows only the real roots among all its complex roots, including the real roots out of the box.  Times means the computing times for the related examples in seconds. We take 5 systems for each case to get their average computing time.
``\textasciitilde{}" means we do not compute the examples. For Bertini, it means the example is out of its computing ability or the computing time is larger than 10000 seconds or their cases are NiDjE. For MK, there are no certified real roots, thus we do not compute the examples. In the examples, we usually set the precision $\rho=10^{-6}$. See the data in Table \ref{tab-ex}.


\newpage

{\scriptsize 
\begin{longtable}[]{@{}llllllllll@{}}
\toprule()
\multirow{2}{*}{Case} & \multirow{2}{*}{Terms} &\multirow{2}{*}{ Coeff} & \multirow{2}{*}{Range} & \multicolumn{3}{|c|}{Roots} & \multicolumn{3}{|c|}{Times(seconds)}\\
\cline{5-7} \cline{8-10}
&&&&Ours& MK& Bertini & Ours& MK& Bertini\\

\midrule()
\endhead

N2D5 & 10 & 10 & -100\textasciitilde100 & 3.0/0/0 & 1.2/11.6 & 3.0 &
0.391 & 0.633 & 0.012 \\
\hline
N2D5E & 10 & 10 & -1\textasciitilde1 & 3.0/0/0 & 1.0/10.6 & \textasciitilde{} & 0.212 & 0.324 & \textasciitilde{} \\
\hline
N2D9 & 10 & 10 & -100\textasciitilde100 & 6.2/49.8/0.2 & 1.8/77.4 & 6.4  & 7.148 & 8.232 & 0.222 \\
\hline
N2D9E & 10 & 10 & -1\textasciitilde1 & 6.4/139.4/0 & 1.8/190.2 & \textasciitilde{} & 13.193 & 14.227 & \textasciitilde{} \\
\hline
N2D51 & 10 & 10 & -100\textasciitilde100 & 7.4/15.6/0 & 0.4/53.4 & 7.8 & 20.862 & 17.265 & 1140.89 \\
\hline
N2D51E & 10 & 10 & -1\textasciitilde1 & 7.6/128.4/0 & 0.4/154.0 & \textasciitilde{} & 38.208 & 38.065 & \textasciitilde{} \\
\hline
N2D101 & 10 & 10 & -100\textasciitilde100 & 6.0/390.0/3.0 & 0.4/421.0 & \textasciitilde{} & 137.141 & 140.961 & \textgreater10000 \\
\hline
N2D101E & 10 & 10 & -1\textasciitilde1 & 5.6/274.4/1.2 & 0.2/290.2 & \textasciitilde{} & 91.37 & 97.727 & \textgreater10000 \\
\hline
N3D9 & 10 & 10 & -100\textasciitilde100 & 12.6/1493.4/0.2 & 0/2348.6 & 13.2 & 721.023 & 798.912 & 5.891 \\
\hline
N3D9E & 10 & 10 & -1\textasciitilde1 & 12.8/519.0/0 & 0/1067.2 & \textasciitilde{} & 173.859 & 207.422 & \textasciitilde{} \\
\hline
N3D51 & 10 & 10 & -2\textasciitilde2 & 5.2/89.4/0 & \textasciitilde{} &
\textasciitilde{} & 45.913 & \textasciitilde{} & \textgreater10000 \\
\hline
N3D101 & 10 & 10 & -2\textasciitilde2 & 8.0/44.2/0 & \textasciitilde{} &
\textasciitilde{} & 4079.618 & \textasciitilde{} & \textgreater10000 \\
\hline
N4D9 & 10 & 10 & -2\textasciitilde2 & 9.8/0/0 & \textasciitilde{} &
\textasciitilde{} & 338.795 & \textasciitilde{} & \textgreater10000 \\
\hline
N4D51 & 10 & 10 & -1\textasciitilde1 & 2.8/0/0 & \textasciitilde{} &
\textasciitilde{} & 131.885 & \textasciitilde{} & \textgreater10000 \\
\hline
N4D101 & 10 & 10 & -1\textasciitilde1 & 1.0/0/0 & \textasciitilde{} &
\textasciitilde{} & 72.618 & \textasciitilde{} & \textgreater10000 \\
\hline
N5D11 & 10 & 10 & -1\textasciitilde1 & 1.2/0/0 & \textasciitilde{} &
\textasciitilde{} & 322.770 & \textasciitilde{} & \textgreater10000 \\
\hline
N6D11 & 10 & 10 & -1\textasciitilde1 & 1.6/0/0 & \textasciitilde{} &
\textasciitilde{} & 8071.526 & \textasciitilde{} & \textgreater10000 \\
\bottomrule()

\caption{Comparing our method with MK and Bertini for polynomial systems.} \label{tab-ex}

\end{longtable}
}
{\noindent Remarks for Table 1}:
\begin{enumerate}
\item In case N2D51, one example has 7 roots and we get 5 of them and miss finding the other 2 from the suspected boxes. In the related case N2D51E, we directly find 7 certified roots for the example but another example miss one root from the suspected boxes.

\item In case N3D9, the number of the total average roots is 13.4. Bertini misses one root, one of whose coordinate is around 3,000,000. Our code finds all the real roots of the systems inside the box $[-100,100]^3$. There are totally 3 roots out of the box.

\item In case N3D101, the computing times for 3 of the 5 systems are less than 100 seconds. But one example takes 20157 seconds and there are 11 certified solutions and no suspected boxes. We find that several roots are very close to each other. One example has 17 certified roots and 162 suspected boxes, which takes 112.398 seconds.

\item In Table \ref{tab-ex}, the examples above N3D51 does not use parallel computing and the results are proved by symbolic computation. From N3D51, we use MPI parallel computing with 30 cores and the results are without proof with symbolic computation.
\end{enumerate}

{\tiny
\begin{table}[thbp]
\center
\begin{tabular}{|c|c|c|c|c|c|}
\hline
System $[g_1,g_2,g_3]$, $g_i=$
 &Terms
 &Times(second) & Roots\\
  \hline
 {\tiny $f_i$} &5 &   1.536  &4.2/0/0 \\
 \hline
{\tiny $f_i*(x_1^2+1)$} &10 &  3.964  &4.2/1.6/0 \\
 \hline
{\tiny $f_i*(x_1^2 + x_2^2 + x_3^2 + 1)$ } &20 & 8.470  &4.2/1.6/0 \\
 \hline
{\tiny $f_i*((x_1 + x_2)^2 + (x_1 + x_3)^2 + 1))$ }&30 & 355.113 &4.2/68.8/0 \\
 \hline
{\tiny $f_i*((x_1 + x_2 + 1)^2 + (x_1 + x_3 - 1)^2 + 1))$} &40 & 326.391   &4.2/166.6/0 \\
 \hline
{\tiny $f_i*((x_1 + x_2 + x_3 + 1)^2 + 1))$}&50 &  4229.643  &4.2/4882.8/0 \\
 \hline
\end{tabular}
\caption{The influence of the number of terms of the input polynomials to our method. The first system $[f_1,f_2,f_3]$ is with randomly generated polynomials of degree 11 and 5 terms. All the systems are expanded before computing. We consider the roots inside $[-1,1]^3$. 
} \label{tab-term}
\end{table}
}

\newpage
In order to analyse the influence of the number of the terms to our method, we design the examples in Table \ref{tab-term}. The systems all have the same solutions but the number of the terms of the polynomials are different. We can find that our method is sensitive to the number of terms of the polynomials in the systems. When the terms of the polynomials increase, the computing times increase and the suspected roots increase. This claim matches the complexity analysis in Lemma \ref{lem-existence}.

The examples in Table \ref{tab-coef} shows the influence of the bit size of the coefficients of the polynomials in the systems to our method. Increasing the bit size of the coefficients does not influence a lot to our method or to MK method but does influence a lot to Bertini since the trick step in Bertini is sensitive to the bit size of the coefficients. This phenomenon matches the bit complexity analysis of our algorithm, see Theorem \ref{thm-complexity}. We can also find that Bertini can not work for systems with large coefficients.

\begin{table}[thbp]
\center
\begin{tabular}{|c|c|c|c|c|c|c|}
\hline
 \multirow{2}{*}{System $[g_1,g_2,g_3]$, $g_i=$}
 & \multicolumn{3}{|c|}{Times(second)} & \multicolumn{3}{|c|}{Roots}\\
  \cline{2-4} \cline{5-7}
&Ours & MK& Bertini &Ours & MK & Bertini\\%
  \hline
 {\tiny $f_i$} &47.093 & 49.262 & 0.213 &4.8/138.2/0 & 0/262.4 & 5.0 \\
 \hline
{\tiny $f_i*2^{10}$} &45.397 & 48.985 & 0.636 &4.8/138.2/0 & 0/262.4 & 5.0 \\
 \hline
{\tiny $f_i*2^{50}$ } &45.362 & 49.254& \textasciitilde{} &4.8/138.2/0 & 0/262.4 & \textasciitilde{}\\
 \hline
{\tiny $f_i*2^{100}$ }&45.464 & 49.302&\textasciitilde{} &4.8/138.2/0 & 0/262.4 & \textasciitilde{}\\
 \hline
{\tiny $f_i*2^{200}$} &45.484 & 49.009 & \textasciitilde{}  &4.8/138.2/0 & 0/262.4 & \textasciitilde{} \\
 \hline
\end{tabular}
\caption{The influence of the bitsize of the coefficients  of the input polynomials to the methods. The first system $[f_1,f_2,f_3]$ is with randomly generated polynomials of degree 5 and 10 terms. The other systems are formed as shown in the table. The systems are expanded before computing. We consider the roots inside $[-100,100]^3$. There is 1 root, one of whose coordinate is out of $[-100,100]$ that is why the number of our roots is 4.8 but that of Bertini is 5.0. Bertini cannot work for systems with large coefficients.} \label{tab-coef}
\end{table}

We also check the influence of the number of real roots inside a box to our method, see Table \ref{tab-root}. If a system has more real roots inside a box, then our methods will take more times. It is reasonable since more roots mean that there are more boxes need to do the existence checking which is time-consuming. The number of real roots of a system almost does not influence the computing times of Bertini since the number of the total complex roots is unchanged under the situation.

\begin{table}[thbp]
\center
\begin{tabular}{|c|c|c|c|c|}
\hline
 \multirow{2}{*}{\# roots}
 & \multicolumn{2}{|c|}{Times} & \multicolumn{2}{|c|}{Roots}\\
 \cline{2-3} \cline{4-5}
 &Bertini & our method &Bertini & our method\\%
  \hline
  8&1.100&  42.839  &8&8/0/0 \\
 \hline
  16& 0.900 &  61.049  &16&16/0/0 \\
 \hline
  24&1.050 &  114.557 &24&24/0/0 \\
 \hline
 32 &1.180&  301.210 &32&32/0/0 \\
 \hline
  40 &0.900& 1082.670 &40&40/728/0 \\
 \hline
 48 &1.280& 982.576 &48&48/536/0 \\
 \hline
\end{tabular}
\caption{The influence of the number of roots of the input polynomials to the methods. The systems are $[f-a_i,g-b_i,h-c_i]$ such that they have 8, 16, 24, 32, 40, 48 real roots respectively, where $[a_i,b_i,c_i]$ are $[62,61,63]$, $[5,61,61]$, $[32,31,37]$, $[22,21,17]$, $[2,63,7]$, $[10,10,10]$ respectively and $f,g,h$ are the expansions of $ \left( 25\,{x}^{2}-2 \right)  \left( 25\,{y}^{2}-11 \right)  \left(
25\,{z}^{2}-5 \right)$,  $\left( 25\,{x}^{2}-11 \right)  \left( 25\,{y}^
{2}-5 \right)  \left( 25\,{z}^{2}-3 \right) , \left( 25\,{x}^{2}-5
 \right)  \left( 25\,{y}^{2}-2 \right)  \left( 25\,{z}^{2}-11 \right)
$ respectively.  All the roots are inside $[-1,1]^3$. } \label{tab-root}
\end{table}

We also check the systems with multiple zeros, see Table \ref{tab-multi}. Usually, systems with multiple real zeros will take more computing time since near the multiple zeros there exist many suspected boxes. To exclude the ones without real roots with the method mentioned in Remark 6 of Algorithm 2 
is time-consuming. And in Case multiN3D12, when excluding suspected boxes, one root is counted twice since it is very close to the boundaries of two suspected boxes.
Some multiple roots of the systems may be computed as several roots by Bertini, so the numbers of roots in Cases multiN2D6 and multiN2D12 are not exactly the numbers of the exact roots of the systems.

\begin{table}[thbp]
\center
\begin{tabular}{|c|c|c|c|c|c|}
\hline
 \multirow{2}{*}{Case}& \multirow{2}{*}{Exact roots}
 & \multicolumn{2}{|c|}{Roots} & \multicolumn{2}{|c|}{Times}\\
 \cline{3-4} \cline{5-6}
& &Ours & Bertini &Ours & Bertini\\%
  \hline
  multiN2D6&3.8 &2.4/8665/1.4 & 6.8 & 334.389 &  0.09 \\
  \hline
multiN2D12&6.4 &2.6/16535/3.8  & 9.8 & 6149.520 &  1.518 \\
  \hline
multiN3D6&6.4 &2.2/6744.2/4.2  & 6.4 & 264.546 &  2.446 \\
  \hline
multiN3D12&6.6 &1.0/16779.2/5.8 &  6.6 & 3115.749  & 1929.558 \\
 \hline
\end{tabular}
\caption{We test systems with multiple real zeros with our method and Bertini. For cases multiN2D6 and multiN2D12, we first randomly generate two polynomials $f$ and $g$ with 4 terms, then the polynomial system $\{f_1,f_2\}$ is got by: $f_1=f*g$ and $f_2=\frac{\partial f_1}{\partial x_2}$.
For cases multiN3D6 and multiN3D12, we first randomly generate two
polynomials $f$ and $g$ with 4 terms, then the polynomial system
$\{f_1,f_2,f_3\}$ is got by: $f_1=f*g, \, f_2=\frac{\partial f_1}{\partial x_1}$ and $f_3$ is a randomly generated polynomial.} \label{tab-multi}
\end{table}

We test the five systems N2D9 in Table \ref{tab-ex}  with our method for different precisions, see Table \ref{tab-pre}. We can find that the higher precision we use, the less suspected boxes we get and the more computing times we need.

\begin{table}[thbp]
\centering
{\small
\begin{tabular}{|c|c|c|c|c|}
\hline
$\rho$ & Certified roots & Suspected boxes & Refined roots & Times \\
\hline
$10^{-2}$ & 1/3/2/3/1 & 31/106/91/87/44 & 2/6/5/2/7 &
1.735/4.346/4.671/3.916/3.192 \\ \hline
$10^{-4}$ & 3/8/6/4/8 & 14/57/128/106/8 & 0/1/1/1/0 &
2.737/8.640/8.849/9.847/4.345 \\ \hline
$10^{-6}$ & 3/9/6/5/8 & 0/26/123/100/0 & 0/0/1/0/0 &
3.014/11.083/14.595/16.017/4.446 \\ \hline
$10^{-8}$ & 3/9/7/5/8 & 0/0/3/0/0 & 0/0/0/0/0 &
3.031/11.389/18.164/17.942/4.335 \\ \hline
$10^{-10}$ & 3/9/7/5/8 & 0/0/0/0/0 & 0/0/0/0/0 &
3.049/11.390/18.332/18.153/4.429 \\ \hline
\end{tabular}
}
\caption{The influence of the termination precisions to the 5 systems in N2D9 in Table \ref{tab-ex}.   }\label{tab-pre}
\end{table}

We check our implementation of parallel computing with MPI on N4D9 in Table \ref{tab-ex}.  We use 4, 8, 16, 32, 60 cores to compute the 5 systems and get the average computing time. The related data is given in Table \ref{tab-parallel}. It shows the speedup of the parallel computing of our code.
\begin{table}[thbp]
\center
\begin{tabular}{|c|c|c|c|c|c|}
\hline
Number of cores & 4 & 8& 16 & 32& 60 \\
\hline
Computing times & 1299.0 & 803.2 & 533.4&364.8&250.8\\
\hline
\end{tabular}
\caption{The influence of the number of cores of parallel computing with our method for N4D9 in Table \ref{tab-ex}.} \label{tab-parallel}
\end{table}

From Table \ref{tab-ex}, we can find that our method is sensitive to the number of variables and the degree of the polynomials also influence the computing times, this suits for the complexity analysis in Theorem \ref{thm-complexity}. For some examples whose roots are distributed in a very bad position, our method may take much time.


Compared to the three related methods, our method is complete for square nonlinear systems with only simple roots, as the MK method. But the interval Newton's method and the $\alpha$-theory method are not complete. Since the exclusion test of MK method and our method are the same, the computing times of two methods are very close if we do not exclude the suspected roots. But we have more certified roots. The existence checking method of MK seldom works for systems with more than two variables. But our method works for systems with 6 variables as shown in Table \ref{tab-ex}. The reason is that Miranda theorem works for evaluating the functions of the systems directly on some of the boundaries of the boxes to require that the result interval does not contain zero. For certifying a simple root, the MK method, our method and the interval Newton's method need only to compute the evaluations of the functions in the system and their order-one derivatives.  The $\alpha$-theory method needs compute higher order derivatives of each function in the system. The existence criteria of the interval Newton method, the $\alpha$-theory method and our method, all need the information of order-one derivatives of the functions. But the Miranda based method needs only the evaluation of the functions. When the box which contains a zero of the system becomes small, the functions are very close to zero. When evaluated on the boundaries of the box, which are also a box (or an interval), the intervals derived from the functions contain zero with high probability, especially for functions with more than two variables. The deeper subdivision does not change the situation and even makes it worse in practice because of the interval calculus. This is the reason why the Miranda based method works only for less variables systems with lower degrees. Our experiments support the claim. Our method avoids using functions directly but using order-one derivatives of the functions for the existence of a zero. It works for systems with more variables and high degrees.

Compared to Bertini, our method works well for polynomial systems with larger B\'{e}zout bound, higher degrees and less variables. Bertini works well for systems with many variables but not so larger B\'{e}zout bound. Our method can find roots of non-polynomial systems but Bertini can not. There are two other advantages of our method: One is that it can find roots of a system in a fixed local region. The other is that each of our certified root box contains exactly one real root of the system. Currently, our method can work only for systems with only several variables. To overcome this shortcoming is our future work.

We will use our code to solve two problems in applications.
\begin{example}
The equations of this example is from robotics and describes the inverse kinematics of an elbow manipulator. One can find the problem in \cite{Hkapur,Hong0}. We solve directly the original one without transforming it to an algebraic system. Thus we have only 6 variables but there are 12 variables in \cite{Hkapur,Hong0}. The following is the system given in \cite{Hkapur,Hong0}. 
{\small
\begin{eqnarray} \label{eqn-ex1}
&&s2\,c5\,s6-s3\,c5\,s6-s4\,c5\,s6+c2\,c6+c3\,c6+c4\,c6-.4077,\nonumber \\
&&c1\,c2\,s5+c1\,c3\,s5+c1\,c4\,s5+s1\,c5-1.9115,\nonumber \\
&& s2\,s5+s3\,s5+s4\,s5-1.9791,\nonumber \\
&& 3\,c1\,c2+2\,c1\,c3+c1\,c4-4.0616,\\
&& 3\,s1\,c2+2\,s1\,c3+s1\,c4-1.7172,\nonumber \\
&&3\,s2+2\,s3+s4-3.9701,\nonumber \\
&& si^2+ci^2-1, i=1,\ldots,6. \nonumber
\end{eqnarray}
}
We replace $si,ci$ with $\sin(6.3\,x_i), \cos(6.3\,x_i)$ for $i=1,\ldots,6$ in \bref{eqn-ex1}. Thus we can consider $x_i$ in $[0,1]$ such that $6.3 x_i$ covers $[0,2\pi]$. Doing so, we get all the solutions of the system. Computing the real roots of the new system in $[0,1]^6$ with precision $\rho=10^{-3}$ and 32 cores with MPI parallel computing, we can get 10 certified roots:
{\tiny
\begin{eqnarray*}
&&[[0.0625 , 0.0644531],[0.0957031 , 0.0976562],[0.148438 , 0.150391],[0.107422 , 0.109375],[0.277344 , 0.279297],[0.230469 , 0.232422]],\\
&&[[0.0625 , 0.0644531],[0.128906 , 0.130859],[0.0820312 , 0.0839844],[0.140625 , 0.142578],[0.275391 , 0.277344],[0.224609 , 0.226562]],\\
&&[[0.0625 , 0.0644531],[0.0957031 , 0.0976562],[0.148438 , 0.150391],[0.107422 , 0.109375],[0.277344 , 0.279297],[0.787109 , 0.789062]],\\
&&[[0.0625 , 0.0644531],[0.128906 , 0.130859],[0.0820312 , 0.0839844],[0.140625 , 0.142578],[0.275391 , 0.277344],[0.783203 , 0.785156]],\\
&&[[0.560547 , 0.5625],[0.400391 , 0.402344],[0.347656 , 0.349609],[0.390625 , 0.392578],[0.220703 , 0.222656],[0.289062 , 0.291016]],\\
&&[[0.560547 , 0.5625],[0.400391 , 0.402344],[0.347656 , 0.349609],[0.390625 , 0.392578],[0.220703 , 0.222656],[0.728516 , 0.730469]],\\
&&[[0.561523 , 0.5625],[0.368164 , 0.369141],[0.415039 , 0.416016],[0.357422 , 0.358398],[0.22168 , 0.222656],[0.285156 , 0.286133]],\\
&&[[0.561523 , 0.5625],[0.390625 , 0.391602],[0.394531 , 0.395508],[0.329102 , 0.330078],[0.304688 , 0.305664],[0.257812 , 0.258789]],\\
&&[[0.561523 , 0.5625],[0.368164 , 0.369141],[0.415039 , 0.416016],[0.357422 , 0.358398],[0.22168 , 0.222656],[0.724609 , 0.725586]],\\
&&[[0.561523 , 0.5625],[0.390625 , 0.391602],[0.394531 , 0.395508],[0.329102 , 0.330078],[0.304688 , 0.305664],[0.693359 , 0.694336]].
\end{eqnarray*}
}
Furthermore, we get 4116 suspected boxes and find none roots from them. We miss 6 real roots.  The computing time is 843.994 seconds.

If we set $\rho=10^{-6}$, we get exactly 16 certified roots and no suspected boxes which takes 2133.05 seconds:
{\tiny
\begin{eqnarray*}
&&[[0.0625 , 0.0644531],[0.0957031 , 0.0976562],[0.148438 , 0.150391],[0.107422 , 0.109375],[0.277344 , 0.279297],[0.230469 , 0.232422]],\\
&&[[0.0625 , 0.0644531],[0.128906 , 0.130859],[0.0820312 , 0.0839844],[0.140625 , 0.142578],[0.275391 , 0.277344],[0.224609 , 0.226562]],\\
&&[[0.0625 , 0.0644531],[0.0957031 , 0.0976562],[0.148438 , 0.150391],[0.107422 , 0.109375],[0.277344 , 0.279297],[0.787109 , 0.789062]],\\
&&[[0.0625 , 0.0644531],[0.128906 , 0.130859],[0.0820312 , 0.0839844],[0.140625 , 0.142578],[0.275391 , 0.277344],[0.783203 , 0.785156]],\\
&&[[0.560547 , 0.5625],[0.400391 , 0.402344],[0.347656 , 0.349609],[0.390625 , 0.392578],[0.220703 , 0.222656],[0.289062 , 0.291016]],\\
&&[[0.560547 , 0.5625],[0.400391 , 0.402344],[0.347656 , 0.349609],[0.390625 , 0.392578],[0.220703 , 0.222656],[0.728516 , 0.730469]],\\
&&[[0.561523 , 0.5625],[0.368164 , 0.369141],[0.415039 , 0.416016],[0.357422 , 0.358398],[0.22168 , 0.222656],[0.285156 , 0.286133]],\\
&&[[0.561523 , 0.5625],[0.390625 , 0.391602],[0.394531 , 0.395508],[0.329102 , 0.330078],[0.304688 , 0.305664],[0.257812 , 0.258789]],\\
&&[[0.561523 , 0.5625],[0.368164 , 0.369141],[0.415039 , 0.416016],[0.357422 , 0.358398],[0.22168 , 0.222656],[0.724609 , 0.725586]],\\
&&[[0.561523 , 0.5625],[0.390625 , 0.391602],[0.394531 , 0.395508],[0.329102 , 0.330078],[0.304688 , 0.305664],[0.693359 , 0.694336]],\\
&&[[0.562012 , 0.5625],[0.40332 , 0.403809],[0.371582 , 0.37207],[0.339844 , 0.340332],[0.308105 , 0.308594],[0.25293 , 0.253418]],\\
&&[[0.562012 , 0.5625],[0.40332 , 0.403809],[0.371582 , 0.37207],[0.339844 , 0.340332],[0.308105 , 0.308594],[0.688965 , 0.689453]],\\
&&[[0.0634766 , 0.0637207],[0.107178 , 0.107422],[0.103516 , 0.10376],[0.169189 , 0.169434],[0.193115 , 0.193359],[0.195068 , 0.195312]],\\
&&[[0.0634766 , 0.0637207],[0.0952148 , 0.095459],[0.126953 , 0.127197],[0.158691 , 0.158936],[0.19043 , 0.190674],[0.190674 , 0.190918]],\\
&&[[0.0634766 , 0.0637207],[0.107178 , 0.107422],[0.103516 , 0.10376],[0.169189 , 0.169434],[0.193115 , 0.193359],[0.756348 , 0.756592]],\\
&&[[0.0634766 , 0.0637207],[0.0952148 , 0.095459],[0.126953 , 0.127197],[0.158691 , 0.158936],[0.19043 , 0.190674],[0.751465 , 0.751709]].
\end{eqnarray*}
}

\end{example}

\begin{example} The following problem is the inverse
position problem for a six-revolute-joint problem in mechanics, one can find it in \cite{Hkapur,morgan}.
The defining equations in the reference are as below.
{\small
\begin{eqnarray*}
&&a_{{{\it i3}}}{\it x_2}\,{\it x_3}+a_{{{\it i4}}}{\it x_2}\,{\it x_4}+a_{{
{\it il}}}{\it x_1}\,{\it x_3}+a_{{{\it i2}}}{\it x_1}\,{\it x_4}+a_{{{
\it i5}}}{\it x_5}\,{\it x_7}+a_{{{\it i6}}}{\it x_5}\,{\it x_8}+a_{{{\it
i7}}}{\it x_6}\,{\it x_7}+a_{{{\it i8}}}{\it x_6}\,{\it x_8}\\
&&+a_{{{\it i9}}
}{\it x_1}+a_{{{\it i10}}}{\it x_2}+a_{{{\it i11}}}{\it x_3}+a_{{{\it i12
}}}{\it x_4}+a_{{{\it i13}}}{\it x_5}+a_{{{\it i14}}}{\it x_6}+a_{{{\it
i15}}}{\it x_7}+a_{{{\it i16}}}{\it x_8}+a_{{{\it i17}}},\\
&&x_i^2+x_{i+1}^2-1, i=1,3,5,7.
\end{eqnarray*}
}
There are 8 equations and 8 variables. The values of $a_{i,j}, j=1,\ldots,17$ can be found in \cite{Hkapur}.
We replace $x_i,x_{i+1}$ with $\sin(6.3\,y_{\frac{i+1}{2}}), \cos(6.3\,y_{\frac{i+1}{2}})$ for $i=1,3,5,7$ in the first equation above.
Then we can get a system with 4 equations and 4 variables. Computing its real roots in $[0,1]^4$ with precision $\rho=10^{-3}$  and 32 cores with MPI parallel computing, we get 9 certified roots:
{\tiny
\begin{eqnarray*}
&&[[0.199219 , 0.203125],[0.195312 , 0.199219],[0.703125 , 0.707031],[0.167969 , 0.171875]],\\
&&[[0.175781 , 0.177734],[0.228516 , 0.230469],[0.375 , 0.376953],[0.609375 , 0.611328]],\\
&&[[0.189453 , 0.191406],[0.871094 , 0.873047],[0.0703125 , 0.0722656],[0.496094 , 0.498047]],\\
&&[[0.917969 , 0.919922],[0.152344 , 0.154297],[0.355469 , 0.357422],[0.630859 , 0.632812]],\\
&&[[0.902344 , 0.904297],[0.0566406 , 0.0585938],[0.662109 , 0.664062],[0.179688 , 0.181641]],\\
&&[[0.859375 , 0.861328],[0.345703 , 0.347656],[0.748047 , 0.75],[0.148438 , 0.150391]],\\
&&[[0.712891 , 0.714844],[0.728516 , 0.730469],[0.746094 , 0.748047],[0.257812 , 0.259766]],\\
&&[[0.0722656 , 0.0732422],[0.379883 , 0.380859],[0.566406 , 0.567383],[0.387695 , 0.388672]],\\
&&[[0.21875 , 0.219727],[0.960938 , 0.961914],[0.0576172 , 0.0585938],[0.513672 , 0.514648]].
\end{eqnarray*}
}
and 2143 suspected boxes. We did not find roots from the suspected boxes. It takes 322.695 seconds.

If we set $\rho=10^{-6}$, we get exactly 12 certified roots and no suspected boxes which takes 902.642 seconds:
{\tiny
\begin{eqnarray*}
&&[[0.199219 , 0.203125],[0.195312 , 0.199219],[0.703125 , 0.707031],[0.167969 , 0.171875]],\\
&&[[0.175781 , 0.177734],[0.228516 , 0.230469],[0.375 , 0.376953],[0.609375 , 0.611328]],\\
&&[[0.189453 , 0.191406],[0.871094 , 0.873047],[0.0703125 , 0.0722656],[0.496094 , 0.498047]],\\
&&[[0.917969 , 0.919922],[0.152344 , 0.154297],[0.355469 , 0.357422],[0.630859 , 0.632812]],\\
&&[[0.902344 , 0.904297],[0.0566406 , 0.0585938],[0.662109 , 0.664062],[0.179688 , 0.181641]],\\
&&[[0.859375 , 0.861328],[0.345703 , 0.347656],[0.748047 , 0.75],[0.148438 , 0.150391]],\\
&&[[0.712891 , 0.714844],[0.728516 , 0.730469],[0.746094 , 0.748047],[0.257812 , 0.259766]],\\
&&[[0.0722656 , 0.0732422],[0.379883 , 0.380859],[0.566406 , 0.567383],[0.387695 , 0.388672]],\\
&&[[0.21875 , 0.219727],[0.960938 , 0.961914],[0.0576172 , 0.0585938],[0.513672 , 0.514648]],\\
&&[[0.0625 , 0.0629883],[0.385742 , 0.38623],[0.567871 , 0.568359],[0.391602 , 0.39209]],\\
&&[[0.0993652 , 0.0996094],[0.712646 , 0.712891],[0.292236 , 0.29248],[0.515625 , 0.515869]],\\
&&[[0.656242 , 0.65625],[0.773285 , 0.773293],[0.321335 , 0.321342],[0.549934 , 0.549942]].\\
\end{eqnarray*}
}

\end{example}
This two examples show that our method works for non-polynomial systems and gives certified solutions.

\section{Conclusion}
In this paper, we propose a numerical method to isolate real zeros of a zero-dimensional multivariate square nonlinear system. We present the concept of the O-M system and the S-M system for a multivariate nonlinear system in an $n$-D box. Based on that, a new existence criterion of a real zero of a system inside a box which is different from Miranda based method is presented, it is much easier to be satisfied than Miranda based method. The uniqueness of a real zero of a system inside a box presented in the paper is related to the existence condition and it contains the traditional Jacobian test. For nonlinear systems, we use the exclusion test to get candidate boxes and for polynomial systems, we can use the bounding polynomials to get the candidate boxes. Then we check the uniqueness and existence conditions for each candidate box. If it succeeds, we get an isolating box of the system. If not, we split these candidate boxes until they satisfy the conditions or their widths reach a given precision. Our method is complete for systems with only finite simple real roots inside a box. We implemented the presented algorithms which shows it works well. The shortcoming of our current method is that we can not solving systems with so many variables.  In the future, we will overcome this shortcoming and consider real zero isolating of high dimensional systems.
\bibliographystyle{plain}

\bibliography{ref_nb_rootfinding}

\end{document}